\documentclass[aps,pra,10pt,a4paper,reprint,superscriptaddress,nofootinbib]{revtex4-2}
\usepackage[utf8]{inputenc}
\usepackage{amssymb}
\usepackage{amsmath}
\usepackage{amsthm}
\usepackage{graphics}
\usepackage{graphicx}
\usepackage{physics}
\usepackage{tikz}
\usepackage[normalem]{ulem}
\usepackage[colorlinks=true,linkcolor=blue,urlcolor=blue,citecolor=blue]{hyperref}
\usepackage{times,txfonts}

\newtheorem{proposition}{Proposition}

\newtheorem{lemma}{Lemma}

\begin{document}
\title{Efficient entanglement-assisted discrimination of a class of many-copy indistinguishable sets}

\author{Abhay Srivastav}
\affiliation{Harish-Chandra Research Institute, Chhatnag Road, Jhunsi, Prayagraj 211019, India}
\affiliation{Optics and Quantum Information Group, The Institute of Mathematical Sciences,\\
C.I.T. Campus, Taramani, Chennai 600113, India}
\affiliation{Homi Bhabha National Institute, Training School Complex, Anushakti Nagar, Mumbai 400094, India}

\author{Saronath Halder}
\email{saronath.halder@gmail.com}
\affiliation{Centre for Quantum Optical Technologies, Centre of New Technologies, University of Warsaw, Banacha 2c, 02-097
Warsaw, Poland} 

\begin{abstract}
We explore entanglement as a resource to distinguish locally indistinguishable orthogonal quantum states. Specifically, we consider sets which contain states from an unextendible product basis along with a pure entangled state. We establish a connection between the aforesaid problem and the entanglement-assisted discrimination of a certain class of many-copy indistinguishable sets. The entanglement-assisted protocols that we construct here are quite efficient, as they render the teleportation-based protocols sub-optimal. In fact, a central aspect of our study is to explore the role of Schmidt rank as a resource to distinguish the states of locally indistinguishable sets. Interestingly, we identify an instance where a set of locally indistinguishable orthogonal states remains locally indistinguishable even with access to any finite number of copies, yet becomes perfectly distinguishable using entangled resources of relatively low cost. This fact makes it possible to compare the degrees of local indistinguishability associated with several locally indistinguishable sets within the same Hilbert space. Consequently, we report a hierarchy of local indistinguishability among the many-copy indistinguishable sets. Thereafter, based on our analysis, we present a theoretical proposal for an information processing protocol exhibiting secure locking of information and its resource-efficient extraction. Furthermore, we also find that the hierarchical difference in local indistinguishability can increase with increasing dimension of the Hilbert space.
\end{abstract}
\maketitle

\section{Introduction}
In quantum information theory, local state discrimination problems are quite well known. To define such a problem, we consider a composite quantum system, the subsystems of which are distributed among spatially separated locations. The parties of those locations are only allowed to use local operations and classical communication (LOCC) \cite{Chitambar14}. However, the composite quantum system is prepared in an unknown state which is secretly taken from a given set. The task is to identify that state via LOCC. If the states of a set are orthogonal then by using a suitable global measurement it is always possible to distinguish the states. Thus, one can identify the state of the system perfectly. On the other hand, perfect discrimination of orthogonal states is not always possible by LOCC. This leads to the notion of local indistinguishability of orthogonal states by LOCC, and the corresponding set is called a locally indistinguishable set. This phenomenon was first observed for a set of two-qutrit pure orthogonal product states \cite{Bennett99-1, Bennett99}. The notion of local indistinguishability only arises for a set consisting of more than two pure states and does not feature in the discrimination of any two orthogonal pure states \cite{Walgate00, Walgate02}. There have been several studies on local indistinguishability and distinguishability of orthogonal product states \cite{DiVincenzo03,Horodecki03,Niset06,Hayashi06,Feng09,Zhang14,Chitambar14-1,Zhang15,Xu16,Croke17,Halder18,Halder19,Rout21,Li22,Shi22} as well as orthogonal entangled states \cite{Ghosh01,Ghosh02,Ghosh04,Fan04,Nathanson05,Watrous05,Owari06,Xin08,Bandyopadhyay10,Bandyopadhyay11,Yu12,Yang13,Cosentino14,Singal16,Singal17,Yuan20,Banik21,Ha23,Liu23}. Ref(s). \cite{Badziag03,Horodecki04} studied the local mutual information of such sets consisting of bipartite states. The locally indistinguishable sets have also been used for quantum data hiding schemes \cite{Rahaman15,Lami18,Lami21,Ha25}. In this work, we consider only orthogonal states, and we assume that the states in a given set are equally probable.

Given a locally indistinguishable set, an interesting research direction is to determine the amount of resource required to distinguish its states by LOCC. Quantum entanglement \cite{Horodecki09} is certainly a very useful resource to distinguish the states of such a set. The role of entanglement as an efficient resource for distinguishing certain locally indistinguishable sets was first examined in Ref. \cite{Cohen08}. The entanglement cost of nonlocal measurements for a particular class of states was subsequently analyzed in Ref. \cite{Bandyopadhyay09}. For bipartite systems, the maximally entangled state serves as a universal resource; however, for most systems involving three or more subsystems it was shown that no such universal resource exists \cite{Bandyopadhyay16}. Subsequent works have explored the role of entanglement in local state discrimination for both product states \cite{Zhang16,Zhang18,Halder18,Rout19,Li19,Li19-1,Shi20,Zhang20,Zhang21,Cao21,Zhang22,Zhou23,Bhunia23,Cao23,Qiao24,Wei24,Zhang25,Cao25} and entangled states \cite{Yu14,Gungor16,Bandyopadhyay18,Shi20-1,Yang20,Bandyopadhyay21,Lovitz22,Bandyopadhyay24,Zhu25,Lim25}. In fact, based on the consumptions of resources, it is possible to put a hierarchy among the degrees of local indistinguishability corresponding to several locally indistinguishable sets. Here, we study entanglement as a resource in local state discrimination problems, and based on this study, we explore a hierarchy of local indistinguishability. The motivation for considering entanglement-assisted local state discrimination problems can be explained as follows. In information processing protocols, we encode classical information in quantum states. To decode this information, we must be able to distinguish among these states. When such discrimination is not possible via LOCC, it becomes natural to consider the use of additional resources, such as entanglement.

Among locally indistinguishable sets, the unextendible product basis (UPB) \cite{Bennett99, DiVincenzo03} is special, as it is associated with different aspects of local indistinguishability \cite{Bennett99-1, Bennett99, Bandyopadhyay11}. The definition of a UPB can be given as the following: Consider a set of orthogonal product states that span a proper subspace of the considered Hilbert space. If the complementary subspace contains no product state, then the given product states form an unextendible product basis. Clearly, the complementary subspace is an entangled subspace. Since a UPB is a locally indistinguishable set \cite{Bennett99, DiVincenzo03, Cohen23}, it is natural to ask how much entanglement is sufficient to distinguish the states of a UPB by LOCC. This question was first addressed in Ref.~\cite{Cohen08}, where it was shown that the states of any two-qutrit UPB can be distinguished by LOCC using a two-qubit maximally entangled state as a resource. In this work, however, we consider a larger set that properly contains a UPB as a subset. We then examine whether a two-qubit maximally entangled state remains sufficient to distinguish the states of this larger set via LOCC. Since the larger set contains at least one entangled state in addition to the UPB states, distinguishing its elements is inherently more complex than discriminating the UPB alone.

The aforesaid problem is important for two key reasons. (i) Since objects like UPBs do not exist in a two-qubit space, the states that we are going to distinguish reside in higher-dimensional spaces. Nevertheless, we investigate their distinguishability using the minimal entangled resource, i.e., a two-qubit maximally entangled state. Clearly, the successful discrimination of the concerned states with such a resource would lead to more efficient protocols, rendering standard teleportation-based protocols sub-optimal. (ii) Furthermore, the aforesaid problem is also connected to the discrimination of states belonging to a class of many-copy indistinguishable sets. Such sets are defined in the following. 

Consider a locally indistinguishable set $\{\rho_1, \rho_2, \dots\}$. We ask whether these states can be distinguished by LOCC when multiple copies are available. That is, we consider the set $\{\rho_1^{\otimes n}, \rho_2^{\otimes n}, \dots\}$, where $n$ is a finite positive integer. If this new set remains locally indistinguishable for all finite $n$, then the original set $\{\rho_1^{\otimes n}, \rho_2^{\otimes n}, \dots\}$ is a locally indistinguishable set in the many-copy scenario. We refer to these sets as many-copy indistinguishable sets. For detailed discussion of such sets, see Refs.~\cite{Bandyopadhyay11, Yu14, Li17, Halder24}. In this work, we introduce a method to distinguish the states of a class of many-copy indistinguishable sets.. More precisely, we consider the set $\{\rho_1, \rho_2, \dots\}\otimes\ket{\xi}$ and ask for which state $\ket{\xi}$, it is possible to distinguish the states $\rho_1, \rho_2, \dots$ perfectly by LOCC. The entanglement-assisted discrimination of many-copy indistinguishable sets remains relatively unexplored. Prior to our work, Ref.~\cite{Yu14} showed that there exists a class of many-copy indistinguishable sets for which perfect local discrimination by LOCC requires an entangled resource state of maximal Schmidt rank. We mention that the Schmidt rank of a pure bipartite entangled state can be defined as the minimum number of product states needed to express that entangled state.

We are now ready to summarize our main findings. First, we explore entanglement as a resource to distinguish locally indistinguishable sets that contain states of a UPB along with a pure entangled state. We then extend our analysis to the entanglement-assisted discrimination of a class of many-copy indistinguishable sets.  Our entanglement-assisted protocols are quite efficient compared to the teleportation-based protocols. In fact, we explore the role of Schmidt rank \cite{Horodecki09} as a resource for local state discrimination. An interesting instance we uncover involves a class of locally indistinguishable sets which cannot be discriminated by LOCC even if any finite number of copies of the states are available. Nevertheless, these states can be perfectly distinguished by LOCC using only minimal entanglement, i.e., a two-qubit entangled state, as a resource. This insight enables us to compare the degrees of local indistinguishability associated with several many-copy indistinguishable sets within the same Hilbert space. Consequently, we establish a hierarchy of local indistinguishability among such sets. Finally, we present a theoretical proposal for an information processing protocol that demonstrates an application of our findings. We also observe that the hierarchical differences in local indistinguishability become more pronounced with increasing system dimension.

\section{Results}
\subsection{Two-qutrit states}
We begin by considering a set that contains two-qutrit orthogonal pure states. Specifically, we focus on the so-called Tiles UPB \cite{Bennett99}. The product states of this UPB along with an entangled state are given by-
\begin{equation}\label{eq1}
\begin{array}{l}
\ket{\Psi_1}  = \ket{0}_A\ket{0-1}_B,~\ket{\Psi_2}  = \ket{0-1}_A\ket{2}_B,\\[1 ex]
\ket{\Psi_3}  = \ket{2}_A\ket{1-2}_B,~\ket{\Psi_4}  = \ket{1-2}_A\ket{0}_B,\\[1 ex]
\ket{\Psi_5}  = \ket{0+1+2}_A\ket{0+1+2}_B,\\[1 ex]
\ket{\Psi_6}  = \ket{0}_A\ket{0+1}_B - \ket{0+1}_A\ket{2}_B,
\end{array}
\end{equation}
where the first subsystem belongs to (A)lice and the second subsystem belongs to (B)ob. Note that $\ket{\Psi_6}$ is orthogonal to all the product states of the UPB and is therefore picked from the entangled subspace. Here, $\ket{u\pm v\pm w}\equiv(1/N)(\ket{u}\pm\ket{v}\pm\ket{w})$, where $N$ is the normalization factor. We do not consider these factors here, as they do not play any role in the discrimination process and can therefore be omitted from the present analysis.

It is well known that the Tiles UPB is a locally indistinguishable set \cite{Bennett99, Fu14, Cohen22}. Consequently, the set of (\ref{eq1}) also remains locally indistinguishable. For this set, we present the following result.

\begin{proposition}\label{prop1}
If a two-qubit maximally entangled state is available as a resource, then the states of (\ref{eq1}) can be perfectly distinguished by LOCC.
\end{proposition}

\begin{proof}
To prove this proposition, it is sufficient to provide a protocol which consumes a two-qubit maximally entangled state (MES) and distinguishes the states of (\ref{eq1}) under LOCC. Before presenting the protocol, we introduce the entangled state $\ket{\phi}_{ab} = \ket{00}_{ab}+\ket{11}_{ab}$ (without a factor for normalization), which is shared between Alice and Bob as a resource. Note that the `$a$' part of $\ket{\phi}_{ab}$ belongs to Alice and the `$b$' part of the same belongs to Bob.

Now, the protocol goes as follows. Bob performs a two-outcome projective measurement where each projection operator is a rank-3 operator. They are given by-
\begin{eqnarray}\label{eq2}
& \mathbb{B}_{1} = & \ket{00}_{Bb}\bra{00}+\ket{10}_{Bb}\bra{10}+\ket{21}_{Bb}\bra{21},\nonumber\\
& \mathbb{B}_{2} = & \ket{01}_{Bb}\bra{01}+\ket{11}_{Bb}\bra{11}+\ket{20}_{Bb}\bra{20}.
\end{eqnarray}
We first discuss the case of $\mathbb{B}_{1}$ and exactly the same steps follow for $\mathbb{B}_{2}$. If $\mathbb{B}_1$ is clicked in the measurement process, then the states of (\ref{eq1}) along with the resource state $\ket{\phi}_{ab}$ are transformed as:
\begin{eqnarray}
&\ket{\Psi_1} \rightarrow & \ket{0}_A\ket{0-1}_B\ket{00}_{ab},\nonumber\\
&\ket{\Psi_2} \rightarrow & \ket{0-1}_A\ket{2}_B\ket{11}_{ab},\nonumber\\
&\ket{\Psi_3} \rightarrow & \ket{21}_{AB}\ket{00}_{ab}-\ket{22}_{AB}\ket{11}_{ab},\nonumber\\
&\ket{\Psi_4} \rightarrow & \ket{1-2}_A\ket{0}_B\ket{00}_{ab},\nonumber\\
&\ket{\Psi_5} \rightarrow & \ket{0+1+2}_A(\ket{0+1}_B\ket{00}_{ab}+\ket{2}_{B}\ket{11}_{ab}),\nonumber\\
&\ket{\Psi_6} \rightarrow & \ket{0}_A\ket{0+1}_B\ket{00}_{ab} - \ket{0+1}_A\ket{2}_B\ket{11}_{ab}.
\end{eqnarray}
The second step of our protocol is Alice making a two-outcome projective measurement where each of the projectors is again a rank-3 operator. They are given by-
\begin{eqnarray}\label{eq4}
& \mathbb{A}_{1} = & \ket{00}_{Aa}\bra{00}+\ket{01}_{Aa}\bra{01}+\ket{11}_{Aa}\bra{11},\nonumber\\
& \mathbb{A}_{2} = & \ket{10}_{Aa}\bra{10}+\ket{20}_{Aa}\bra{20}+\ket{21}_{Aa}\bra{21}.
\end{eqnarray}
First, we consider that $\mathbb{A}_1$ is clicked in the measurement process, then the states are transformed as:
\begin{eqnarray}
&\ket{\Psi_1} \rightarrow & \ket{0}_A\ket{0-1}_B\ket{00}_{ab},\nonumber\\
&\ket{\Psi_2} \rightarrow & \ket{0-1}_A\ket{2}_B\ket{11}_{ab},\nonumber\\
&\ket{\Psi_5} \rightarrow & \ket{0}_A\ket{0+1}_B\ket{00}_{ab}+\ket{0+1}_{A}\ket{2}_{B}\ket{11}_{ab},\nonumber\\
&\ket{\Psi_6} \rightarrow & \ket{0}_A\ket{0+1}_B\ket{00}_{ab} - \ket{0+1}_A\ket{2}_B\ket{11}_{ab}.
\end{eqnarray}
On these states, the protocol proceeds as follows. Alice again makes a two-outcome projective measurement to distinguish between two subspaces, on Alice's side one is spanned by $\ket{0-1}_A\ket{1}_a$ and the other is spanned by \{$\ket{0}_A\ket{0}_a$, $\ket{0+1}_A\ket{1}_a$\}. If the first subspace is detected then the state is $\ket{\Psi_2}$, otherwise, Alice and Bob are left with other three states of the above equation, given by-
\begin{eqnarray}
&\ket{\Psi_1} \rightarrow & \ket{0}_A\ket{0-1}_B\ket{00}_{ab},\nonumber\\
&\ket{\Psi_5} \rightarrow & \ket{0}_A\ket{0+1}_B\ket{00}_{ab}+\ket{0+1}_{A}\ket{2}_{B}\ket{11}_{ab},\nonumber\\
&\ket{\Psi_6} \rightarrow & \ket{0}_A\ket{0+1}_B\ket{00}_{ab} - \ket{0+1}_A\ket{2}_B\ket{11}_{ab}.
\end{eqnarray}
Next, Bob does a two-outcome projective measurement to distinguish between two subspaces, on Bob's side one is spanned by $\ket{0-1}_B\ket{0}_b$ and the other is spanned by \{$\ket{0+1}_B\ket{0}_b$, $\ket{2}_B\ket{1}_b$\}. If the first subspace is detected then the state is $\ket{\Psi_1}$, otherwise, they are left with the other two states of the above equation. These two states are orthogonal and can be perfectly distinguished by LOCC as described in Ref.\cite{Walgate00}. 

Going back to (\ref{eq4}), we describe what if $\mathbb{A}_2$ is clicked in the measurement process. It eliminates $\ket{\Psi_1}$, $\ket{\Psi_2}$, and $\ket{\Psi_6}$ completely, and the other states are transformed as:
\begin{eqnarray}
&\ket{\Psi_3}\rightarrow &\ket{21}_{AB}\ket{00}_{ab}-\ket{22}_{AB}\ket{11}_{ab}, \nonumber\\
&\ket{\Psi_4}\rightarrow & \ket{1-2}_A\ket{0}_B\ket{00}_{ab}, \nonumber\\
&\ket{\Psi_5}\rightarrow & \ket{1+2}_A\ket{0+1}_B\ket{00}_{ab}+\ket{2}_A\ket{2}_{B}\ket{11}_{ab}.
\end{eqnarray}
On these states, Bob makes a two-outcome projective measurement to distinguish between two subspaces, on Bob's side one is spanned by $\ket{00}_{Bb}$ and the other is spanned by \{$\ket{10}_{Bb}$, $\ket{21}_{Bb}$\}. On detection of any of these subspaces, Alice and Bob are left with two orthogonal pure states, which can always be distinguished perfectly by LOCC \cite{Walgate00}. This completes the first part of our protocol.

For the second part of our protocol, we go back to (\ref{eq2}) and assume that in the measurement process $\mathbb{B}_2$ is clicked. In this case, similar steps follow as with $\mathbb{B}_1$, as previously described.
\end{proof}

Note that in (\ref{eq1}), we consider a particular entangled state. Here, we emphasize that if the entangled state is changed, then the protocol must be modified accordingly. For example, if we consider the entangled state as $\ket{1+2}\ket{0}-\ket{2}\ket{1+2}$, then Bob must start with a measurement defined by the following projection operators:
\begin{equation*}
\begin{array}{c}
\mathbb{B}_{1} = \ket{00}_{Bb}\bra{00}+\ket{11}_{Bb}\bra{11}+\ket{21}_{Bb}\bra{21},\\
\mathbb{B}_{2} = \ket{01}_{Bb}\bra{01}+\ket{10}_{Bb}\bra{10}+\ket{20}_{Bb}\bra{20}.
\end{array}
\end{equation*}
However, it is currently unknown whether a proposition similar to Proposition \ref{prop1} can be established for any entangled state chosen from the entangled subspace corresponding to the Tiles UPB. Moreover, the above proposition opens up a new direction of research, giving rise to many additional questions. Some of these questions are as follows: (a) How can our protocol be generalized when more than one entangled state is present? (b) In such a case, will a two-qubit maximally entangled state still be effective? (c) What happens if we consider a UPB other than the Tiles UPB? While it is interesting to study these questions, the single instance reported through Proposition \ref{prop1} is sufficient to establish our result regarding entanglement-assisted discrimination of a class of many-copy indistinguishable sets. In the following, we develop this result step by step. 

{\it Set containing a mixed state}.-- Next, we consider a state discrimination problem involving two orthogonal quantum states—one being a mixed state and the other a pure state. The mixed state is the normalized projection operator onto the space spanned by the states of UPB given in (\ref{eq1}), while the pure state is the entangled state also given in the same equation. We denote the mixed state by $\rho(=1/N^\prime\sum_{i=1}^5\ketbra{\Psi_i}{\Psi_i})$ and the pure state by $\ket{\psi}(=\ket{\Psi_6})$; $N^\prime$ is a coefficient for normalization. It is known that the set $\{\rho^{\otimes n}, \ket{\psi}^{\otimes n}\}$ cannot be distinguished by LOCC for any finite $n$ \cite{Bandyopadhyay11}. We now turn to entanglement-assisted discrimination, where the states to be distinguished belong to the set $\{\rho, \ket{\psi}\}$. The goal is to determine whether these states can be perfectly distinguished by LOCC with the aid of entanglement, while consuming minimal entanglement resources. Before addressing this problem, we present a sufficient condition for distinguishing mixed states.

\begin{lemma}\label{suffcond}
Two orthogonal mixed states are perfectly distinguishable by a protocol if the states in the spectral decomposition of the mixed states are perfectly distinguishable by the same protocol. 
\end{lemma}

\begin{proof}
Consider two orthogonal mixed states $\rho_1$ and $\rho_2$ acting on some Hilbert space $\mathcal{H}$ such that the spectral decompositions of the states are given by- $\rho_1=\sum_{i}b_{i}\ketbra{\phi_{i}}{\phi_{i}}$ and $\rho_2=\sum_{i}c_{i}\ketbra{\psi_{i}}{\psi_{i}}$, respectively, with $\langle\phi_i|\phi_j\rangle = \delta_{ij} = \langle\psi_{i}|\psi_{j}\rangle$ and $\langle\phi_{i}|\psi_{j}\rangle=0$. 

Let us now consider a locally implementable measurement $\mathcal{M}$, defined by a set of positive operator valued measure (POVM) elements $\{M^i_{\phi}, M^i_{\psi}\}$, $\sum_iM^i_{\phi}+M^i_{\psi}=\mathbb{I}$, such that $\mathcal{M}$ distinguishes the states $\{\ket{\phi_i}\}\cup\{\ket{\psi_i}\}$ perfectly. So, $\langle\phi_i|M^j_{\phi}|\phi_i\rangle$ = $\delta_{ij}$ = $\langle\psi_i|M^j_{\psi}|\psi_i\rangle$ and $\langle\phi_i|M^j_{\psi}|\phi_i\rangle$ = 0 = $\langle\psi_i|M^j_{\phi}|\psi_i\rangle~\forall~i, j$. These imply that $\mbox{Tr}[\rho_2\sum_iM^i_{\phi}]$ = 0 = $\mbox{Tr}[\rho_1\sum_iM^i_{\psi}]$. Moreover, $\sum_i M^i_{\phi} + M^i_{\psi} = \mathbb{I}$, therefore, $\mbox{Tr}[\rho_1\sum_iM^i_{\phi}]$ = 1 = $\mbox{Tr}[\rho_2\sum_iM^i_{\psi}]$. Clearly, $\mathcal{M}$ also distinguishes the mixed states $\rho_1$ and $\rho_2$ perfectly. This completes the proof. 
\end{proof}

We mention that the sufficient condition provided in the above lemma may not be a necessary condition for distinguishing orthogonal mixed states. Now, we present the following proposition for entanglement-assisted discrimination of the set $\{\rho, \ket{\psi}\}$.  

\begin{proposition}\label{prop2}
The states of the set $\mathcal{S}=\{\rho, \ket{\psi}\}$ can be perfectly distinguished by LOCC if a two-qubit maximally entangled state is available as a resource. 
\end{proposition}

\begin{proof}
Based on the sufficient condition given in Lemma~\ref{suffcond} for distinguishing mixed states, we know that the key quantity of interest is the amount of entanglement required to distinguish the constituent pure states in the decomposition of $\rho$ together with the state $\ket{\psi}$. This task ultimately boils down to distinguishing the states in (\ref{eq1}), where first five states comprise $\rho$ and the sixth state is $\ket{\psi}$. Furthermore, Proposition \ref{prop1} has already established that these states are perfectly distinguishable by LOCC given a two-qubit MES as a resource. This concludes the proof.
\end{proof}

\noindent
In the above proposition, the state $\rho$ is a normalized projection operator on the subspace spanned by the states of UPB. However, this proposition is true for any rank-5 mixed state $\rho^\prime$ which is a convex combination of the states of the UPB. Next, we discuss some consequences that follow from Propositions \ref{prop1} and \ref{prop2}.

\begin{itemize}
\item The set $\{\rho, \ket{\psi}\}$ is a two-qutrit set, but to distinguish its states, a two-qutrit entangled state is not necessary. Clearly, any teleportation-based protocol to distinguish $\rho$ and $\ket{\psi}$ is sub-optimal. The same conclusion also applies to the states of (\ref{eq1}).

\item The set $\{\rho, \ket{\psi}\}$ is {\it strange} -- it remains locally indistinguishable regardless of how many finite copies of the states are available. However, the states of the set are distinguishable by LOCC when a lower dimensional (minimum dimensional) entangled state is available as a resource.

\item Thus, the local indistinguishability of $\{\rho, \ket{\psi}\}$ is robust against having many copies of the states as a resource, but becomes relatively weak when entanglement is available as a resource.

\item The connection we made between Propositions \ref{prop1} and \ref{prop2} is useful to shed light on the entanglement-assisted discrimination of a class of many-copy indistinguishable sets.
\end{itemize}

{\it A hierarchy}.-- We exhibit the hierarchy by considering several locally indistinguishable sets within the same Hilbert space. All of these locally indistinguishable sets are many-copy indistinguishable. However, under entanglement-assisted discrimination scenario, to distinguish the states of these sets, we may need different amounts of resources. Mathematically, we define this hierarchy as follows. Suppose, we consider two distinct classes of many-copy indistinguishable sets. An example for a set from one of these classes is given in Proposition \ref{prop2}, which we denote by $\mathcal{S}$. Example for the other class is denoted by $\mathcal{S}'$. $\mathcal{S}'$ is defined as follows: it also contains only two states like $\mathcal{S}$, one is pure and the other is mixed. The pure state is a maximally entangled state in a given Hilbert space, while the mixed state is a normalized projector onto the complementary subspace of the maximally entangled state. Such sets can be found in \cite{Yu14}. We next consider a quantity $E$, the value of which cannot be increased under LOCC. In the setting of entanglement-assisted discrimination, $E$ can be taken as an entanglement monotone \cite{Vidal00}. Suppose that after discriminating the sets $\mathcal{S}$ and $\mathcal{S}'$ using entanglement-assisted LOCC protocols, we find that $E(\mathcal{S}')>E(\mathcal{S})$, where $E(\mathcal{S}')$ is the value of the given entanglement monotone necessary for the discrimination of the set $\mathcal{S}'$, and $E(\mathcal{S})$ is the value of the given entanglement monotone sufficient for the discrimination of the set $\mathcal{S}$. Then, we say that there is a hierarchy between $\mathcal{S}$ and $\mathcal{S}'$. In particular, $\mathcal{S}'$ is more locally indistinguishable than $\mathcal{S}$.

Previously, we have shown that the two-qutrit states of the set $\mathcal{S}$ = $\{\rho, \ket{\psi}\}$ are perfectly distinguishable by LOCC if both parties have access to a two-qubit MES; see Proposition \ref{prop2} in this regard. Now, consider an entangled state $\ket{\phi}$ in the same Hilbert space. This state is a two-qutrit maximally entangled state. Let us consider another state, $\sigma=\frac{1}{8}(\mathbb{I}-\ketbra{\phi})$, $\mathbb{I}$ being the identity operator onto the two-qutrit Hilbert space. It is known that the set $\mathcal{S}^\prime$ = $\{\sigma, \ket{\phi}\}$ is also locally indistinguishable under many-copy scenario \cite{Yu14}. Moreover, to distinguish the states of this set by LOCC, it is necessary to consume an entangled state as resource whose Schmidt rank is at least three \cite{Yu14}. 

So, for two-qutrit systems, we now have two types of sets $\mathcal{S}$ and $\mathcal{S}^\prime$, both of them are locally indistinguishable in many-copy scenario. However, to distinguish the states of $\mathcal{S}$ under entanglement-assisted local state discrimination scenario, it is sufficient to have a two-qubit MES, the Schmidt rank of which is no greater than two. On the other hand, to distinguish the states of $\mathcal{S}^\prime$ under entanglement-assisted local state discrimination scenario, we need an entangled state as resource whose Schmidt rank is no less than three. Now, the Schmidt rank of a bipartite entangled state is a resource, as it cannot be increased under LOCC \cite{Lo01}. Furthermore, it is possible to provide an entanglement measure generalizing the concept of Schmidt rank \cite{Eisert01}. So, what we obtain here is $E(\mathcal{S}')>E(\mathcal{S})$ under entanglement-assisted local state discrimination scenario; $E(\mathcal{S}')$ is Schmidt rank three, which is necessary for the discrimination of the two-qutrit states of the set $\mathcal{S}'$, and $E(\mathcal{S})$ is Schmidt rank two, which is sufficient for the discrimination of the two-qutrit states of the set $\mathcal{S}$. This establishes the hierarchy: the set $\mathcal{S}'$ is more locally indistinguishable than the set $\mathcal{S}$. In the following, we provide some consequences of our findings. 

\begin{itemize}
\item For one locally indistinguishable set ($\mathcal{S'}$), to distinguish its states by LOCC, we need an entangled state from higher dimensional Hilbert space. Interestingly, there exists another locally indistinguishable set ($\mathcal{S})$, to distinguish its states we need an entangled state from lower dimensional Hilbert space. Thus, we exhibit an instance where we have analyzed the role of Schmidt rank in local state discrimination problems. 

\item By using a measurement which is more powerful than LOCC, one may try to exhibit a similar hierarchy. However, such a measurement may not be physically motivated. Thus, to implement these measurements under a class of physically motivated operations like LOCC, we ultimately need entanglement.  
\end{itemize}

It is natural to ask if the analysis that we have presented here is limited to the case of a two-qutrit system or if similar instances can also be found in higher dimensional Hilbert spaces. To address this question, we consider a two-ququad Tiles UPB and show that such an instance can indeed be found in a higher dimensional Hilbert space. Interestingly, we also show that the difference in the degrees of local indistinguishability for different many-copy indistinguishable sets can increase in higher dimensions. 

\subsection{Two-ququad states}
In the following, we first provide an orthogonal product basis, based on which the UPB can be constructed. 

\begin{widetext}
\small
\begin{equation}\label{prodbasisquad}
\begin{array}{llll}
\ket{\Psi_1^{(1)}} = \ket{0}_A\ket{0+1+2}_B, & \ket{\Psi_1^{(2)}} = \ket{0}_A\ket{0+\omega1+\omega^{2}2}_B, & \ket{\Psi_1^{(3)}} = \ket{0}_A\ket{0+\omega^{2}1+\omega2}_B, & \ket{\Psi_2^{(1)}} = \ket{0+1+2}_A\ket{3}_B,\\[1 ex] \ket{\Psi_2^{(2)}} = \ket{0+\omega1+\omega^{2}2}_A\ket{3}_B, & \ket{\Psi_2^{(3)}} = \ket{0+\omega^{2}1+\omega2}_A\ket{3}_B, & \ket{\Psi_3^{(1)}} = \ket{3}_A\ket{1+2+3}_B, & \ket{\Psi_3^{(2)}} = \ket{3}_A\ket{1+\omega2+\omega^{2}3}_B,\\[1 ex] \ket{\Psi_3^{(3)}} = \ket{3}_A\ket{1+\omega^{2}2+\omega3}_B, & \ket{\Psi_4^{(1)}} = \ket{1+2+3}_A\ket{0}_B, & \ket{\Psi_4^{(2)}} = \ket{1+\omega2+\omega^{2}3}_A\ket{0}_B, & \ket{\Psi_4^{(3)}} = \ket{1+\omega^{2}2+\omega3}_A\ket{0}_B,\\[1 ex] \ket{\Psi_5^{(1)}} = \ket{1+2}_A\ket{1+2}_B, & \ket{\Psi_5^{(2)}} = \ket{1+2}_A\ket{1-2}_B, & \ket{\Psi_5^{(3)}} = \ket{1-2}_A\ket{1+2}_B, & \ket{\Psi_5^{(4)}} = \ket{1-2}_A\ket{1-2}_B.
\end{array}
\end{equation}
\end{widetext}
\normalsize
where $\ket{0+\omega1+\omega^{2}2}$ stands for $\frac{1}{\sqrt{3}}(\ket{0}+\omega\ket{1}+\omega^{2}\ket{2})$ etc., and $\omega$ is the cube root of unity. Let us now consider a product state: 
\begin{equation}
\ket{\Psi_6}= \ket{0+1+2+3}_A\ket{0+1+2+3}_B.
\end{equation}
Now, $\mathcal{U}=\mathcal{S}\cup\{\ket{\Psi_6}\}\setminus\{\ket{\Psi_{i}^{(1)}}\}_{i=1}^5$ is a UPB in a two-ququad system ($4 \otimes 4$), $\mathcal{S}$ is the product basis given in Eq.\eqref{prodbasisquad}. The proof that $\mathcal{U}$ is a UPB follows from \cite{Shi20}. Let us now consider an entangled state $\ket{\Psi_7}\in\mathcal{U^{\perp}}$ as
\begin{equation}
\ket{\Psi_7} = \ket{0}_A\ket{0+1+2}_B-\ket{0+1+2}_A\ket{3}_B. 
\end{equation}
Then, for the set $\mathcal{U} \bigcup \{\ket{\Psi_7}\}$, we prove the following:

\begin{proposition}\label{prop3}
If a two-qubit MES is available as a resource, then the states of the set $\mathcal{U} \bigcup \{\ket{\Psi_7}\}$ can be perfectly distinguished by LOCC.
\end{proposition}

\begin{proof}
Similar to the proof of Proposition \ref{prop1}, we introduce an ancillary system $\ket{\phi}_{ab} = \ket{00}_{ab} + \ket{11}_{ab}$ (without normalization), where Alice and Bob hold a qubit each. Our protocol then goes as follows, Bob performs a two-outcome projective measurement, where each measurement operator is a rank-4 projector given as:
\small
\begin{equation}
\begin{array}{c}
\mathbb{B}_{1} =  \ket{01}_{Bb}\bra{01}+\ket{11}_{Bb}\bra{11}+\ket{21}_{Bb}\bra{21}+\ket{30}_{Bb}\bra{30},\\[1 ex]

\mathbb{B}_{2}  =  \ket{00}_{Bb}\bra{00}+\ket{10}_{Bb}\bra{10}+\ket{20}_{Bb}\bra{20}+\ket{31}_{Bb}\bra{31},
\end{array}
\end{equation}
\normalsize
where we see that $\mathbb{B}_{2}$ differs from $\mathbb{B}_{1}$ only in the state of the ancillary system ($0_{b}\leftrightarrow1_{b})$. So, we discuss only the case of $\mathbb{B}_{1}$ and a similar method follows for $\mathbb{B}_{2}$. If $\mathbb{B}_{1}$ is clicked in the measurement process, then the states in the set $\mathcal{U}\bigcup\{\ket{\Psi}_{7}\}$ together with the resource state $\ket{\phi}_{ab}$ transform as:
\begin{widetext}
\begin{equation}
\begin{array}{ll}
\ket{\Psi_1^{(2)}} \rightarrow  \ket{0}_A\ket{0+\omega1+\omega^{2}2}_B\ket{11}_{ab}, & \ket{\Psi_1^{(3)}} \rightarrow  \ket{0}_A\ket{0+\omega^{2}1+\omega2}_B\ket{11}_{ab}, \\[1 ex] 

\ket{\Psi_2^{(2)}} \rightarrow  \ket{0+\omega1+\omega^{2}2}_A\ket{3}_B\ket{00}_{ab}, & \ket{\Psi_2^{(3)}} \rightarrow  \ket{0+\omega^{2}1+\omega2}_A\ket{3}_B\ket{00}_{ab}, \\[1 ex] 

\ket{\Psi_3^{(2)}} \rightarrow  \ket{3}_A\left(\ket{1+\omega2}_B\ket{11}_{ab}+\ket{\omega^{2}3}_B\ket{00}_{ab}\right), & \ket{\Psi_3^{(3)}} \rightarrow  \ket{3}_A\left(\ket{1+\omega^{2}2}_B\ket{11}_{ab}+\ket{\omega3}_B\ket{00}_{ab}\right), \\[1 ex]

\ket{\Psi_4^{(2)}} \rightarrow  \ket{1+\omega2+\omega^{2}3}_A\ket{0}_B\ket{11}_{ab}, & \ket{\Psi_4^{(3)}} \rightarrow  \ket{1+\omega^{2}2+\omega3}_A\ket{0}_B\ket{11}_{ab}, \\[1 ex] 

\ket{\Psi_5^{(2)}} \rightarrow  \ket{1+2}_A\ket{1-2}_B\ket{11}_{ab}, & \ket{\Psi_5^{(3)}} \rightarrow \ket{1-2}_A\ket{1+2}_B\ket{11}_{ab}, \\[1 ex]

\ket{\Psi_5^{(4)}} \rightarrow \ket{1-2}_A\ket{1-2}_B\ket{11}_{ab}, & \ket{\Psi_6} \rightarrow  \ket{0+1+2+3}_A\left(\ket{0+1+2}_B\ket{11}_{ab}+\ket{3}_{B}\ket{00}_{ab}\right), \\[1 ex] 

\ket{\Psi_7} \rightarrow \ket{0}_A\ket{0+1+2}_B\ket{11}_{ab}-\ket{0+1+2}_A\ket{3}_{B}\ket{00}_{ab}. &
\end{array}
\end{equation}
\end{widetext}
\normalsize
The second step of our protocol is Alice making a three-outcome measurement with two measurement operators being rank-2 projectors and the third one being "everything else",
\begin{align}
&\mathbb{A}_{1}  = \ket{00+\omega10+\omega^{2}20}_{Aa}\bra{00+\omega10+\omega^{2}20}&\nonumber\\
&\hspace{.75cm} +\ket{00+\omega^{2}10+\omega20}_{Aa}\bra{00+\omega^{2}10+\omega20},\nonumber\\
&\mathbb{A}_{2}  = \ket{00+10+20}_{Aa}\bra{00+10+20}+\ket{01}_{Aa}\bra{01},\nonumber\\
&\mathbb{A}_{3}  = \ket{30}_{Aa}\bra{30}+\ket{31}_{Aa}\bra{31}+\ket{11}_{Aa}\bra{11}+\ket{21}_{Aa}\bra{21}.\nonumber
\end{align}
Let us first consider that $\mathbb{A}_{1}$ is clicked in the measurement process. Then, the only states that do not get eliminated are $\ket{\Psi_2^{(2)}}$ and $\ket{\Psi_2^{(3)}}$, which remain unchanged. These two states are orthogonal on Alice's side who can make a local measurement to distinguish them. We now consider the case where $\mathbb{A}_{2}$ is clicked in the measurement process. It eliminates $\ket{\Psi_2^{(2)}}$, $\ket{\Psi_2^{(3)}}$, $\ket{\Psi_3^{(2)}}$, $\ket{\Psi_3^{(3)}}$, $\ket{\Psi_4^{(2)}}$, $\ket{\Psi_4^{(3)}}$, $\ket{\Psi_5^{(2)}}$, $\ket{\Psi_5^{(3)}}$, and $\ket{\Psi_5^{(4)}}$ completely, and the remaining states transform as:
\begin{eqnarray}
&\ket{\Psi_1^{(2)}}  \rightarrow & \ket{0}_A\ket{0+\omega1+\omega^{2}2}_B\ket{11}_{ab},\nonumber\\
&\ket{\Psi_1^{(3)}}  \rightarrow & \ket{0}_A\ket{0+\omega^{2}1+\omega2}_B\ket{11}_{ab},\nonumber\\
&\ket{\Psi_6}  \;\;      \rightarrow & \ket{0+1+2}_A\ket{3}_B\ket{00}_{ab}+\ket{0}_A\ket{0+1+2}_B\ket{11}_{ab},\nonumber\\
&\ket{\Psi_7}   \;\;     \rightarrow & \ket{0}_A\ket{0+1+2}_B\ket{11}_{ab}-\ket{0+1+2}_A\ket{3}_B\ket{00}_{ab}.\nonumber
\end{eqnarray}
On these states, Bob again makes a two-outcome projective measurement to distinguish between two subspaces, on Bob's side one is spanned by $\{\ket{01+\omega11+\omega^{2}21}_{Bb}, \ket{01+\omega^{2}11+\omega21}_{Bb}\}$, and the other is spanned by $\{\ket{01+11+21}_{Bb}, \ket{30}_{Bb}\}$. On the detection of either of these two subspaces, Alice and Bob are left with two orthogonal pure states which can always be distinguished perfectly by LOCC \cite{Walgate00}. Let us now consider the third and the last outcome of Alice's measurement, $\mathbb{A}_{3}$. It eliminates $\ket{\Psi_1^{(2)}}$, $\ket{\Psi_1^{(3)}}$, $\ket{\Psi_2^{(2)}}$, $\ket{\Psi_2^{(3)}}$, and $\ket{\Psi_7}$ completely. Then, the remaining states transform as:
\begin{eqnarray}
&\ket{\Psi_3^{(2)}}  \rightarrow & \ket{3}_A(\ket{1+\omega2}_B\ket{11}_{ab}+\ket{\omega^{2}3}_B\ket{00}_{ab}),\nonumber\\
&\ket{\Psi_3^{(3)}}  \rightarrow & \ket{3}_A(\ket{1+\omega^{2}2}_B\ket{11}_{ab}+\ket{\omega3}_B\ket{00}_{ab}),\nonumber\\
&\ket{\Psi_4^{(2)}}  \rightarrow & \ket{1+\omega2+\omega^{2}3}_A\ket{0}_B\ket{11}_{ab},\nonumber\\
&\ket{\Psi_4^{(3)}}  \rightarrow & \ket{1+\omega^{2}2+\omega3}_A\ket{0}_B\ket{11}_{ab},\nonumber\\
&\ket{\Psi_5^{(2)}}  \rightarrow & \ket{1+2}_A\ket{1-2}_B\ket{11}_{ab},\nonumber\\
&\ket{\Psi_5^{(3)}}  \rightarrow & \ket{1-2}_A\ket{1+2}_B\ket{11}_{ab},\nonumber\\
&\ket{\Psi_5^{(4)}}  \rightarrow & \ket{1-2}_A\ket{1-2}_B\ket{11}_{ab},\nonumber\\
&\ket{\Psi_6}  \;\;      \rightarrow & \ket{1+2+3}_A\ket{0+1+2}_B\ket{11}_{ab}+\ket{3}_{A}\ket{3}_{B}\ket{00}_{ab}.\nonumber
\end{eqnarray}
On these states, Bob makes a two-outcome projective measurement to distinguish between two subspaces, on Bob's side one is spanned by $\{\ket{01}_{Bb}\}$, and the other is spanned by $\{\ket{11}_{Bb}, \ket{21}_{Bb}, \ket{30}_{Bb}\}$. If the first subspace is detected, then Alice and Bob are left with the following three states:
\begin{eqnarray}
&\ket{\Psi_4^{(2)}}  \rightarrow & \ket{1+\omega2+\omega^{2}3}_A\ket{0}_B\ket{11}_{ab},\nonumber\\
&\ket{\Psi_4^{(3)}}  \rightarrow & \ket{1+\omega^{2}2+\omega3}_A\ket{0}_B\ket{11}_{ab},\nonumber\\
&\ket{\Psi_6}  \;\;      \rightarrow & \ket{1+2+3}_A\ket{0}_B\ket{11}_{ab}.\nonumber
\end{eqnarray}
These three states are orthogonal on Alice's side who can make a local measurement to distinguish them perfectly. If the second subspace is detected, then the states $\ket{\Psi_4^{(2)}}$ and $\ket{\Psi_4^{(3)}}$ get eliminated, and the remaining states transform as:
\begin{eqnarray}
&\ket{\Psi_3^{(2)}}  \rightarrow & \ket{3}_A(\ket{1+\omega2}_B\ket{11}_{ab}+\ket{\omega^{2}3}_B\ket{00}_{ab}),\nonumber\\
&\ket{\Psi_3^{(3)}}  \rightarrow & \ket{3}_A(\ket{1+\omega^{2}2}_B\ket{11}_{ab}+\ket{\omega3}_B\ket{00}_{ab}),\nonumber\\
&\ket{\Psi_5^{(2)}}  \rightarrow & \ket{1+2}_A\ket{1-2}_B\ket{11}_{ab},\nonumber\\
&\ket{\Psi_5^{(3)}}  \rightarrow & \ket{1-2}_A\ket{1+2}_B\ket{11}_{ab},\nonumber\\
&\ket{\Psi_5^{(4)}}  \rightarrow & \ket{1-2}_A\ket{1-2}_B\ket{11}_{ab},\nonumber\\
&\ket{\Psi_6}  \;\;      \rightarrow & \ket{1+2+3}_A\ket{1+2}_B\ket{11}_{ab}+\ket{3}_A\ket{3}_B\ket{00}_{ab}.\nonumber
\end{eqnarray}
On these six states, Alice makes a three-outcome projective measurement to distinguish between three subspaces, on Alice's side one is spanned by $\{\ket{1
+2}\ket{1}_{Aa}\}$, second by $\{\ket{1-2}\ket{1}_{Aa}\}$, and the third by $\{\ket{30}_{Aa}, \ket{31}_{Aa}\}$. If the first or the second subspace is detected, then Alice and Bob are left with two orthogonal pure states which can always be perfectly distinguished by LOCC \cite{Walgate00}. If the third subspace is detected, then the states $\ket{\Psi_5^{(2)}}$, $\ket{\Psi_5^{(3)}}$, and $\ket{\Psi_5^{(4)}}$ get eliminated completely, and the remaining states transform as:
\begin{eqnarray}
&\ket{\Psi_3^{(2)}}  \rightarrow & \ket{3}_A(\ket{1+\omega2}_B\ket{11}_{ab}+\ket{\omega^{2}3}_B\ket{00}_{ab}),\nonumber\\
&\ket{\Psi_3^{(3)}}  \rightarrow & \ket{3}_A(\ket{1+\omega^{2}2}_B\ket{11}_{ab}+\ket{\omega3}_B\ket{00}_{ab}),\nonumber\\
&\ket{\Psi_6}  \;\;      \rightarrow & \ket{3}_A(\ket{1+2}_B\ket{11}_{ab}+\ket{3}_B\ket{00}_{ab}).\nonumber
\end{eqnarray}
On these states, Alice makes a measurement on her ancillary system in $\{\ket{0\pm1}_a\}$ basis, depending on the outcome of the measurement the states transform as:
\begin{eqnarray}
&\ket{\Psi_3^{(2)}}  \rightarrow & \ket{3}_A\ket{0\pm1}_a(\ket{1+\omega2}_B\ket{1}_{b}\pm\ket{\omega^{2}3}_B\ket{0}_{b}),\nonumber\\
&\ket{\Psi_3^{(3)}}  \rightarrow & \ket{3}_A\ket{0\pm1}_a(\ket{1+\omega^{2}2}_B\ket{1}_{b}\pm\ket{\omega3}_B\ket{0}_{b}),\nonumber\\
&\ket{\Psi_6}  \;\;      \rightarrow & \ket{3}_A\ket{0\pm1}_a(\ket{1+2}_B\ket{1}_{ab}\pm\ket{3}_B\ket{0}_{b}).\nonumber
\end{eqnarray}
These states are now orthogonal on Bob's side, and therefore locally distinguishable. This completes the first part of the protocol. 

For the second part, we assume that $\mathbb{B}_{2}$ is clicked in the measurement process, then similar steps follow as mentioned before. This concludes the proof.
\end{proof}

Building on Proposition \ref{prop2} and the related discussions, and with the aid of Proposition \ref{prop3}, we construct a two-ququad many-copy indistinguishable set whose states can be perfectly distinguished by LOCC when a two-qubit MES is available as a resource. This set, denoted by $\mathbb{S}$, consists of two states: a mixed state given by the normalized projection operator onto the UPB $\mathcal{U}$, and a pure entangled state $\ket{\Psi_7}$. We now proceed to discuss several consequences, outlined below.

\begin{itemize}
\item For both sets, $\mathcal{U} \bigcup \{\ket{\Psi_7}\}$ and $\mathbb{S}$, any teleportation-based protocol to distinguish the constituent quantum states is sub-optimal.

\item It is natural to compare between the sets which belong to the same Hilbert space. Therefore, we consider another set $\mathbb{S}^\prime$, which contains two states. The first one is a two-ququad MES, and the second one is a mixed state which is a normalized projecter onto the orthogonal complement of the two-ququad MES. $\mathbb{S}^\prime$ is known to be locally indistinguishable in many-copy scenario \cite{Yu14}. Again, to distinguish the states of this set under entanglement-assisted local state discrimination setting, it is necessary to have an entangled state of Schmidt rank four \cite{Yu14}. Clearly, compared to this set, for the discrimination of the states of $\mathbb{S}$, we need less resource, in terms of Schmidt rank, when we consider entanglement-assisted local state discrimination setting.

\item The dimension is increased from $\mathcal{S}$ to $\mathbb{S}$ but still the resource state remains the same, i.e., a two-qubit MES is sufficient is to distinguish the states of both sets by LOCC. So, $\Delta E' = [E(\mathbb{S}')-E(\mathbb{S})]>\Delta E = [E(\mathcal{S}')-E(\mathcal{S})]$; where $E(\mathbb{S}')/E(\mathcal{S}')$ is the necessary Schmidt rank for the entanglement-assisted discrimination of $\mathbb{S}'/\mathcal{S}'$ and $E(\mathbb{S})/E(\mathcal{S})$ is the sufficient Schmidt rank for the entanglement-assisted discrimination of $\mathbb{S}/\mathcal{S}$. Clearly, in this case, the difference in degrees of local indistinguishability is increased with increasing dimension.
\end{itemize}

In the following, we present a theoretical model of an information processing protocol that demonstrates an application of our results. In fact, we highlight how the hierarchy established in our analysis plays a useful role in this protocol.

\subsection{Application}
We consider a scenario in which a manager wishes to distribute a piece of two-level classical information securely among spatially separated parties. The information should be distributed in such a way that the distant parties cannot extract the information completely using LOCC. However, if required, there must be a way to fully extract the information using less resource. We refer to this protocol as {\it secure locking of information and its resource-efficient extraction}. In this regard, one can see Ref.~\cite{Goswami23}.

\begin{figure}[t!]
\centering
\includegraphics[scale=0.33]{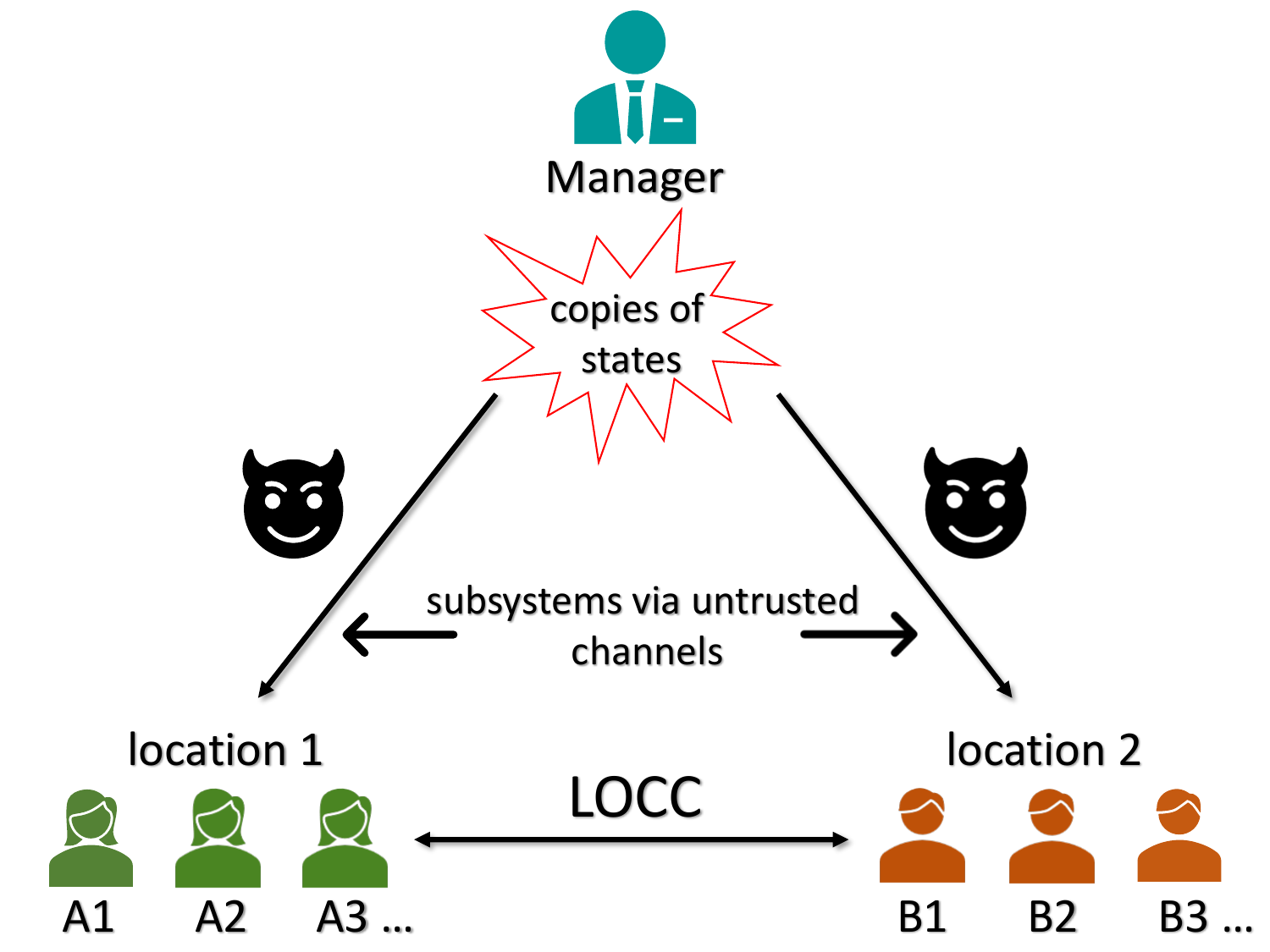}
\caption{Schematic diagram for our information processing protocol.}\label{fig1}
\end{figure}

However, there are other conditions to be addressed in this protocol. We consider two different locations. At both locations, there are multiple observers. At location 1, we consider A1, A2, A3, etc. These observers, though belong to a single location, are disjoint, i.e., they do not pass any information to each other. One may consider here that different `$\mbox{A}i$'s are present at different times. Similarly, at location 2, we consider disjoint observers B1, B2, B3, etc. The Manager wishes to distribute the same piece of two-level classical information between each of the pairs $\{\mbox{A}i, \mbox{B}i\}$, $\forall~i=1,2,3,\dots$, where `$|i|$' is finite. Each pair $(\mbox{A}i,\mbox{B}i)$ is restricted to use only LOCC. The lab of the Manager is connected with both the locations via quantum channels. High dimensional quantum system is allowed to be sent through each of these channels, but each channel is allowed to be used only once. Furthermore, these channels cannot be trusted because there are eavesdroppers acting on both channels. For a schematic diagram of the protocol, see Fig.~\ref{fig1}.

To achieve the goal, the manager proceeds as follows. (S)he produces identical copies of the states of some set, say, $\{\sigma_1^{\otimes n}, \sigma_2^{\otimes n}\}$; where $n$ is an integer equal to the number of pairs $\{\mbox{A}i,~ \mbox{B}i\}$, i.e., $|i|$. The states $\sigma_1$ and $\sigma_2$ are orthogonal to each other. This ensures that complete information extraction is possible. Moreover, the set $\{\sigma_1, \sigma_2\}$ is locally indistinguishable, so that the two-level classical information encoded in the states cannot be extracted via LOCC. The subsystems corresponding to $\{\sigma_1^{\otimes n}, \sigma_2^{\otimes n}\}$
are then transmitted to locations 1 and 2 through a one-time use of the quantum channels. However, since the channels are not trusted, local eavesdroppers can have access to multiple identical copies of the states. Hence, if they use LOCC, they might be successful in information extraction, if the set $\{\sigma_1, \sigma_2\}$ is locally distinguishable in the many-copy limit. We note that this constraint is not purely hypothetical. Recently, perfect discrimination of quantum states, using many copies, was demonstrated experimentally \cite{Zhang23}. However, to overcome this constraint, one can use the set of Proposition \ref{prop2}. This set also provides an additional measure of security as the states of this set cannot be distinguished by LOCC with non-zero probability in a conclusive manner, even if finitely many-copies are available \cite{Bandyopadhyay11}. Now that we have addressed all the constraints of the protocol, we now discuss the extraction of the information by the pairs $(\mbox{A}i,\mbox{B}i)$ using less resource. By Proposition \ref{prop2}, it is established that each $(\mbox{A}i,\mbox{B}i)$ will be able to fully extract the encoded information if they share a two-qubit MES as resource.

Furthermore, as we have seen, there exist other many-copy indistinguishable sets that can be distinguished using entanglement-assisted discrimination. Hence, the manager can also use these other sets (instead of the set of Proposition \ref{prop2}) to execute the protocol. However, for each of these sets, entanglement is the resource to extract the encoded information, and it is always better to use lesser amount of a resource. Now through the hierarchy related discussions, we have already established the following: to distinguish the states of the set of Proposition \ref{prop2} by LOCC, it is sufficient to use lesser resource compared to other sets. In particular, we have talked about the sets $\mathcal{S}$ and $\mathcal{S^\prime}$. So, for the protocol described above, we ask which set is appropriate for achieving the goal of the protocol when both $\mathcal{S}$ and $\mathcal{S^\prime}$ are available. We argue that $\mathcal{S}$ is more appropriate for the protocol, as the information encoded in the states of $\mathcal{S}$ can be decoded by LOCC consuming an entangled state of lesser Schmidt rank. Thus, the goal of the protocol can be achieved efficiently using the set of Proposition \ref{prop2}.

\section{Conclusion}
We explored several new local state discrimination problems. In particular, we considered sets containing the product states of a UPB together with an entangled state and determined the amount of entanglement sufficient to distinguish such states. We then connected this problem to the discrimination of states that remain locally indistinguishable even in the many-copy scenario. These analyses shed light on the role of Schmidt rank as a resource in local state discrimination problems. Specifically, we presented two many-copy indistinguishable sets $S'$ and $S$ for two-qutrit systems. We found that the entangled state necessary to discriminate $S'$ must have Schmidt rank at least three, where as a state with Schmidt rank two is sufficient to discriminate the set $S$. Therefore, as a consequence of these studies, we demonstrated the existence of a hierarchy of local indistinguishability among many-copy indistinguishable sets. We also generalized this result for two-ququad systems. Finally, we provided a theoretical proposal for an information processing protocol where our analysis can be applied. In the results section, we also highlighted several open questions for future investigation.

\section*{Acknowledgments}
A. Srivastav acknowledges the support of the INFOSYS scholarship during the major part of this work. S. Halder was supported with a postdoc fellowship from Harish-Chandra Research Institute, Prayagraj (Allahabad) when this work was started. Then, S. Halder was supported by the ``Quantum Optical Technologies'' project, carried out within the International Research Agendas programme of the Foundation for Polish Science co-financed by the European Union under the European Regional Development Fund.

\bibliographystyle{apsrev4-2}
\bibliography{ref}

\begin{thebibliography}{86}%
\makeatletter
\providecommand \@ifxundefined [1]{%
 \@ifx{#1\undefined}
}%
\providecommand \@ifnum [1]{%
 \ifnum #1\expandafter \@firstoftwo
 \else \expandafter \@secondoftwo
 \fi
}%
\providecommand \@ifx [1]{%
 \ifx #1\expandafter \@firstoftwo
 \else \expandafter \@secondoftwo
 \fi
}%
\providecommand \natexlab [1]{#1}%
\providecommand \enquote  [1]{``#1''}%
\providecommand \bibnamefont  [1]{#1}%
\providecommand \bibfnamefont [1]{#1}%
\providecommand \citenamefont [1]{#1}%
\providecommand \href@noop [0]{\@secondoftwo}%
\providecommand \href [0]{\begingroup \@sanitize@url \@href}%
\providecommand \@href[1]{\@@startlink{#1}\@@href}%
\providecommand \@@href[1]{\endgroup#1\@@endlink}%
\providecommand \@sanitize@url [0]{\catcode `\\12\catcode `\$12\catcode
  `\&12\catcode `\#12\catcode `\^12\catcode `\_12\catcode `\%12\relax}%
\providecommand \@@startlink[1]{}%
\providecommand \@@endlink[0]{}%
\providecommand \url  [0]{\begingroup\@sanitize@url \@url }%
\providecommand \@url [1]{\endgroup\@href {#1}{\urlprefix }}%
\providecommand \urlprefix  [0]{URL }%
\providecommand \Eprint [0]{\href }%
\providecommand \doibase [0]{https://doi.org/}%
\providecommand \selectlanguage [0]{\@gobble}%
\providecommand \bibinfo  [0]{\@secondoftwo}%
\providecommand \bibfield  [0]{\@secondoftwo}%
\providecommand \translation [1]{[#1]}%
\providecommand \BibitemOpen [0]{}%
\providecommand \bibitemStop [0]{}%
\providecommand \bibitemNoStop [0]{.\EOS\space}%
\providecommand \EOS [0]{\spacefactor3000\relax}%
\providecommand \BibitemShut  [1]{\csname bibitem#1\endcsname}%
\let\auto@bib@innerbib\@empty
\bibitem [{\citenamefont {Chitambar}\ \emph
  {et~al.}(2014{\natexlab{a}})\citenamefont {Chitambar}, \citenamefont {Leung},
  \citenamefont {Man\v{c}inska}, \citenamefont {Ozols},\ and\ \citenamefont
  {Winter}}]{Chitambar14}%
  \BibitemOpen
  \bibfield  {author} {\bibinfo {author} {\bibfnamefont {E.}~\bibnamefont
  {Chitambar}}, \bibinfo {author} {\bibfnamefont {D.}~\bibnamefont {Leung}},
  \bibinfo {author} {\bibfnamefont {L.}~\bibnamefont {Man\v{c}inska}}, \bibinfo
  {author} {\bibfnamefont {M.}~\bibnamefont {Ozols}},\ and\ \bibinfo {author}
  {\bibfnamefont {A.}~\bibnamefont {Winter}},\ }\href
  {https://doi.org/https://doi.org/10.1007/s00220-014-1953-9} {\bibfield
  {journal} {\bibinfo  {journal} {Commun. Math. Phys.}\ }\textbf {\bibinfo
  {volume} {328}},\ \bibinfo {pages} {303} (\bibinfo {year}
  {2014}{\natexlab{a}})}\BibitemShut {NoStop}%
\bibitem [{\citenamefont {Bennett}\ \emph
  {et~al.}(1999{\natexlab{a}})\citenamefont {Bennett}, \citenamefont
  {DiVincenzo}, \citenamefont {Fuchs}, \citenamefont {Mor}, \citenamefont
  {Rains}, \citenamefont {Shor}, \citenamefont {Smolin},\ and\ \citenamefont
  {Wootters}}]{Bennett99-1}%
  \BibitemOpen
  \bibfield  {author} {\bibinfo {author} {\bibfnamefont {C.~H.}\ \bibnamefont
  {Bennett}}, \bibinfo {author} {\bibfnamefont {D.~P.}\ \bibnamefont
  {DiVincenzo}}, \bibinfo {author} {\bibfnamefont {C.~A.}\ \bibnamefont
  {Fuchs}}, \bibinfo {author} {\bibfnamefont {T.}~\bibnamefont {Mor}}, \bibinfo
  {author} {\bibfnamefont {E.}~\bibnamefont {Rains}}, \bibinfo {author}
  {\bibfnamefont {P.~W.}\ \bibnamefont {Shor}}, \bibinfo {author}
  {\bibfnamefont {J.~A.}\ \bibnamefont {Smolin}},\ and\ \bibinfo {author}
  {\bibfnamefont {W.~K.}\ \bibnamefont {Wootters}},\ }\href
  {https://doi.org/10.1103/PhysRevA.59.1070} {\bibfield  {journal} {\bibinfo
  {journal} {Phys. Rev. A}\ }\textbf {\bibinfo {volume} {59}},\ \bibinfo
  {pages} {1070} (\bibinfo {year} {1999}{\natexlab{a}})}\BibitemShut {NoStop}%
\bibitem [{\citenamefont {Bennett}\ \emph
  {et~al.}(1999{\natexlab{b}})\citenamefont {Bennett}, \citenamefont
  {DiVincenzo}, \citenamefont {Mor}, \citenamefont {Shor}, \citenamefont
  {Smolin},\ and\ \citenamefont {Terhal}}]{Bennett99}%
  \BibitemOpen
  \bibfield  {author} {\bibinfo {author} {\bibfnamefont {C.~H.}\ \bibnamefont
  {Bennett}}, \bibinfo {author} {\bibfnamefont {D.~P.}\ \bibnamefont
  {DiVincenzo}}, \bibinfo {author} {\bibfnamefont {T.}~\bibnamefont {Mor}},
  \bibinfo {author} {\bibfnamefont {P.~W.}\ \bibnamefont {Shor}}, \bibinfo
  {author} {\bibfnamefont {J.~A.}\ \bibnamefont {Smolin}},\ and\ \bibinfo
  {author} {\bibfnamefont {B.~M.}\ \bibnamefont {Terhal}},\ }\href
  {https://doi.org/10.1103/PhysRevLett.82.5385} {\bibfield  {journal} {\bibinfo
   {journal} {Phys. Rev. Lett.}\ }\textbf {\bibinfo {volume} {82}},\ \bibinfo
  {pages} {5385} (\bibinfo {year} {1999}{\natexlab{b}})}\BibitemShut {NoStop}%
\bibitem [{\citenamefont {Walgate}\ \emph {et~al.}(2000)\citenamefont
  {Walgate}, \citenamefont {Short}, \citenamefont {Hardy},\ and\ \citenamefont
  {Vedral}}]{Walgate00}%
  \BibitemOpen
  \bibfield  {author} {\bibinfo {author} {\bibfnamefont {J.}~\bibnamefont
  {Walgate}}, \bibinfo {author} {\bibfnamefont {A.~J.}\ \bibnamefont {Short}},
  \bibinfo {author} {\bibfnamefont {L.}~\bibnamefont {Hardy}},\ and\ \bibinfo
  {author} {\bibfnamefont {V.}~\bibnamefont {Vedral}},\ }\href
  {https://doi.org/10.1103/PhysRevLett.85.4972} {\bibfield  {journal} {\bibinfo
   {journal} {Phys. Rev. Lett.}\ }\textbf {\bibinfo {volume} {85}},\ \bibinfo
  {pages} {4972} (\bibinfo {year} {2000})}\BibitemShut {NoStop}%
\bibitem [{\citenamefont {Walgate}\ and\ \citenamefont
  {Hardy}(2002)}]{Walgate02}%
  \BibitemOpen
  \bibfield  {author} {\bibinfo {author} {\bibfnamefont {J.}~\bibnamefont
  {Walgate}}\ and\ \bibinfo {author} {\bibfnamefont {L.}~\bibnamefont
  {Hardy}},\ }\href {https://doi.org/10.1103/PhysRevLett.89.147901} {\bibfield
  {journal} {\bibinfo  {journal} {Phys. Rev. Lett.}\ }\textbf {\bibinfo
  {volume} {89}},\ \bibinfo {pages} {147901} (\bibinfo {year}
  {2002})}\BibitemShut {NoStop}%
\bibitem [{\citenamefont {DiVincenzo}\ \emph {et~al.}(2003)\citenamefont
  {DiVincenzo}, \citenamefont {Mor}, \citenamefont {Shor}, \citenamefont
  {Smolin},\ and\ \citenamefont {Terhal}}]{DiVincenzo03}%
  \BibitemOpen
  \bibfield  {author} {\bibinfo {author} {\bibfnamefont {D.~P.}\ \bibnamefont
  {DiVincenzo}}, \bibinfo {author} {\bibfnamefont {T.}~\bibnamefont {Mor}},
  \bibinfo {author} {\bibfnamefont {P.~W.}\ \bibnamefont {Shor}}, \bibinfo
  {author} {\bibfnamefont {J.~A.}\ \bibnamefont {Smolin}},\ and\ \bibinfo
  {author} {\bibfnamefont {B.~M.}\ \bibnamefont {Terhal}},\ }\href
  {https://doi.org/10.1007/s00220-003-0877-6} {\bibfield  {journal} {\bibinfo
  {journal} {Commun. Math. Phys.}\ }\textbf {\bibinfo {volume} {238}},\
  \bibinfo {pages} {379} (\bibinfo {year} {2003})}\BibitemShut {NoStop}%
\bibitem [{\citenamefont {Horodecki}\ \emph {et~al.}(2003)\citenamefont
  {Horodecki}, \citenamefont {Sen(De)}, \citenamefont {Sen},\ and\
  \citenamefont {Horodecki}}]{Horodecki03}%
  \BibitemOpen
  \bibfield  {author} {\bibinfo {author} {\bibfnamefont {M.}~\bibnamefont
  {Horodecki}}, \bibinfo {author} {\bibfnamefont {A.}~\bibnamefont {Sen(De)}},
  \bibinfo {author} {\bibfnamefont {U.}~\bibnamefont {Sen}},\ and\ \bibinfo
  {author} {\bibfnamefont {K.}~\bibnamefont {Horodecki}},\ }\href
  {https://doi.org/10.1103/PhysRevLett.90.047902} {\bibfield  {journal}
  {\bibinfo  {journal} {Phys. Rev. Lett.}\ }\textbf {\bibinfo {volume} {90}},\
  \bibinfo {pages} {047902} (\bibinfo {year} {2003})}\BibitemShut {NoStop}%
\bibitem [{\citenamefont {Niset}\ and\ \citenamefont {Cerf}(2006)}]{Niset06}%
  \BibitemOpen
  \bibfield  {author} {\bibinfo {author} {\bibfnamefont {J.}~\bibnamefont
  {Niset}}\ and\ \bibinfo {author} {\bibfnamefont {N.~J.}\ \bibnamefont
  {Cerf}},\ }\href {https://doi.org/10.1103/PhysRevA.74.052103} {\bibfield
  {journal} {\bibinfo  {journal} {Phys. Rev. A}\ }\textbf {\bibinfo {volume}
  {74}},\ \bibinfo {pages} {052103} (\bibinfo {year} {2006})}\BibitemShut
  {NoStop}%
\bibitem [{\citenamefont {Hayashi}\ \emph {et~al.}(2006)\citenamefont
  {Hayashi}, \citenamefont {Markham}, \citenamefont {Murao}, \citenamefont
  {Owari},\ and\ \citenamefont {Virmani}}]{Hayashi06}%
  \BibitemOpen
  \bibfield  {author} {\bibinfo {author} {\bibfnamefont {M.}~\bibnamefont
  {Hayashi}}, \bibinfo {author} {\bibfnamefont {D.}~\bibnamefont {Markham}},
  \bibinfo {author} {\bibfnamefont {M.}~\bibnamefont {Murao}}, \bibinfo
  {author} {\bibfnamefont {M.}~\bibnamefont {Owari}},\ and\ \bibinfo {author}
  {\bibfnamefont {S.}~\bibnamefont {Virmani}},\ }\href
  {https://doi.org/10.1103/PhysRevLett.96.040501} {\bibfield  {journal}
  {\bibinfo  {journal} {Phys. Rev. Lett.}\ }\textbf {\bibinfo {volume} {96}},\
  \bibinfo {pages} {040501} (\bibinfo {year} {2006})}\BibitemShut {NoStop}%
\bibitem [{\citenamefont {{Feng}}\ and\ \citenamefont {{Shi}}(2009)}]{Feng09}%
  \BibitemOpen
  \bibfield  {author} {\bibinfo {author} {\bibfnamefont {Y.}~\bibnamefont
  {{Feng}}}\ and\ \bibinfo {author} {\bibfnamefont {Y.}~\bibnamefont {{Shi}}},\
  }\href {https://doi.org/10.1109/TIT.2009.2018330} {\bibfield  {journal}
  {\bibinfo  {journal} {IEEE Trans. Inf. Theory}\ }\textbf {\bibinfo {volume}
  {55}},\ \bibinfo {pages} {2799} (\bibinfo {year} {2009})}\BibitemShut
  {NoStop}%
\bibitem [{\citenamefont {Zhang}\ \emph {et~al.}(2014)\citenamefont {Zhang},
  \citenamefont {Gao}, \citenamefont {Tian}, \citenamefont {Cao},\ and\
  \citenamefont {Wen}}]{Zhang14}%
  \BibitemOpen
  \bibfield  {author} {\bibinfo {author} {\bibfnamefont {Z.-C.}\ \bibnamefont
  {Zhang}}, \bibinfo {author} {\bibfnamefont {F.}~\bibnamefont {Gao}}, \bibinfo
  {author} {\bibfnamefont {G.-J.}\ \bibnamefont {Tian}}, \bibinfo {author}
  {\bibfnamefont {T.-Q.}\ \bibnamefont {Cao}},\ and\ \bibinfo {author}
  {\bibfnamefont {Q.-Y.}\ \bibnamefont {Wen}},\ }\href
  {https://doi.org/10.1103/PhysRevA.90.022313} {\bibfield  {journal} {\bibinfo
  {journal} {Phys. Rev. A}\ }\textbf {\bibinfo {volume} {90}},\ \bibinfo
  {pages} {022313} (\bibinfo {year} {2014})}\BibitemShut {NoStop}%
\bibitem [{\citenamefont {Chitambar}\ \emph
  {et~al.}(2014{\natexlab{b}})\citenamefont {Chitambar}, \citenamefont {Duan},\
  and\ \citenamefont {Hsieh}}]{Chitambar14-1}%
  \BibitemOpen
  \bibfield  {author} {\bibinfo {author} {\bibfnamefont {E.}~\bibnamefont
  {Chitambar}}, \bibinfo {author} {\bibfnamefont {R.}~\bibnamefont {Duan}},\
  and\ \bibinfo {author} {\bibfnamefont {M.-H.}\ \bibnamefont {Hsieh}},\ }\href
  {https://doi.org/10.1109/TIT.2013.2295356} {\bibfield  {journal} {\bibinfo
  {journal} {IEEE Trans. Inf. Theory}\ }\textbf {\bibinfo {volume} {60}},\
  \bibinfo {pages} {1549} (\bibinfo {year} {2014}{\natexlab{b}})}\BibitemShut
  {NoStop}%
\bibitem [{\citenamefont {Zhang}\ \emph {et~al.}(2015)\citenamefont {Zhang},
  \citenamefont {Gao}, \citenamefont {Qin}, \citenamefont {Yang},\ and\
  \citenamefont {Wen}}]{Zhang15}%
  \BibitemOpen
  \bibfield  {author} {\bibinfo {author} {\bibfnamefont {Z.-C.}\ \bibnamefont
  {Zhang}}, \bibinfo {author} {\bibfnamefont {F.}~\bibnamefont {Gao}}, \bibinfo
  {author} {\bibfnamefont {S.-J.}\ \bibnamefont {Qin}}, \bibinfo {author}
  {\bibfnamefont {Y.-H.}\ \bibnamefont {Yang}},\ and\ \bibinfo {author}
  {\bibfnamefont {Q.-Y.}\ \bibnamefont {Wen}},\ }\href
  {https://doi.org/10.1103/PhysRevA.92.012332} {\bibfield  {journal} {\bibinfo
  {journal} {Phys. Rev. A}\ }\textbf {\bibinfo {volume} {92}},\ \bibinfo
  {pages} {012332} (\bibinfo {year} {2015})}\BibitemShut {NoStop}%
\bibitem [{\citenamefont {Xu}\ \emph {et~al.}(2016)\citenamefont {Xu},
  \citenamefont {Wen}, \citenamefont {Qin}, \citenamefont {Yang},\ and\
  \citenamefont {Gao}}]{Xu16}%
  \BibitemOpen
  \bibfield  {author} {\bibinfo {author} {\bibfnamefont {G.-B.}\ \bibnamefont
  {Xu}}, \bibinfo {author} {\bibfnamefont {Q.-Y.}\ \bibnamefont {Wen}},
  \bibinfo {author} {\bibfnamefont {S.-J.}\ \bibnamefont {Qin}}, \bibinfo
  {author} {\bibfnamefont {Y.-H.}\ \bibnamefont {Yang}},\ and\ \bibinfo
  {author} {\bibfnamefont {F.}~\bibnamefont {Gao}},\ }\href
  {https://doi.org/10.1103/PhysRevA.93.032341} {\bibfield  {journal} {\bibinfo
  {journal} {Phys. Rev. A}\ }\textbf {\bibinfo {volume} {93}},\ \bibinfo
  {pages} {032341} (\bibinfo {year} {2016})}\BibitemShut {NoStop}%
\bibitem [{\citenamefont {Croke}\ and\ \citenamefont
  {Barnett}(2017)}]{Croke17}%
  \BibitemOpen
  \bibfield  {author} {\bibinfo {author} {\bibfnamefont {S.}~\bibnamefont
  {Croke}}\ and\ \bibinfo {author} {\bibfnamefont {S.~M.}\ \bibnamefont
  {Barnett}},\ }\href {https://doi.org/10.1103/PhysRevA.95.012337} {\bibfield
  {journal} {\bibinfo  {journal} {Phys. Rev. A}\ }\textbf {\bibinfo {volume}
  {95}},\ \bibinfo {pages} {012337} (\bibinfo {year} {2017})}\BibitemShut
  {NoStop}%
\bibitem [{\citenamefont {Halder}(2018)}]{Halder18}%
  \BibitemOpen
  \bibfield  {author} {\bibinfo {author} {\bibfnamefont {S.}~\bibnamefont
  {Halder}},\ }\href {https://doi.org/10.1103/PhysRevA.98.022303} {\bibfield
  {journal} {\bibinfo  {journal} {Phys. Rev. A}\ }\textbf {\bibinfo {volume}
  {98}},\ \bibinfo {pages} {022303} (\bibinfo {year} {2018})}\BibitemShut
  {NoStop}%
\bibitem [{\citenamefont {Halder}\ \emph {et~al.}(2019)\citenamefont {Halder},
  \citenamefont {Banik}, \citenamefont {Agrawal},\ and\ \citenamefont
  {Bandyopadhyay}}]{Halder19}%
  \BibitemOpen
  \bibfield  {author} {\bibinfo {author} {\bibfnamefont {S.}~\bibnamefont
  {Halder}}, \bibinfo {author} {\bibfnamefont {M.}~\bibnamefont {Banik}},
  \bibinfo {author} {\bibfnamefont {S.}~\bibnamefont {Agrawal}},\ and\ \bibinfo
  {author} {\bibfnamefont {S.}~\bibnamefont {Bandyopadhyay}},\ }\href
  {https://doi.org/10.1103/PhysRevLett.122.040403} {\bibfield  {journal}
  {\bibinfo  {journal} {Phys. Rev. Lett.}\ }\textbf {\bibinfo {volume} {122}},\
  \bibinfo {pages} {040403} (\bibinfo {year} {2019})}\BibitemShut {NoStop}%
\bibitem [{\citenamefont {Rout}\ \emph {et~al.}(2021)\citenamefont {Rout},
  \citenamefont {Maity}, \citenamefont {Mukherjee}, \citenamefont {Halder},\
  and\ \citenamefont {Banik}}]{Rout21}%
  \BibitemOpen
  \bibfield  {author} {\bibinfo {author} {\bibfnamefont {S.}~\bibnamefont
  {Rout}}, \bibinfo {author} {\bibfnamefont {A.~G.}\ \bibnamefont {Maity}},
  \bibinfo {author} {\bibfnamefont {A.}~\bibnamefont {Mukherjee}}, \bibinfo
  {author} {\bibfnamefont {S.}~\bibnamefont {Halder}},\ and\ \bibinfo {author}
  {\bibfnamefont {M.}~\bibnamefont {Banik}},\ }\href
  {https://doi.org/10.1103/PhysRevA.104.052433} {\bibfield  {journal} {\bibinfo
   {journal} {Phys. Rev. A}\ }\textbf {\bibinfo {volume} {104}},\ \bibinfo
  {pages} {052433} (\bibinfo {year} {2021})}\BibitemShut {NoStop}%
\bibitem [{\citenamefont {Li}\ and\ \citenamefont {Zheng}(2022)}]{Li22}%
  \BibitemOpen
  \bibfield  {author} {\bibinfo {author} {\bibfnamefont {M.-S.}\ \bibnamefont
  {Li}}\ and\ \bibinfo {author} {\bibfnamefont {Z.-J.}\ \bibnamefont {Zheng}},\
  }\href {https://doi.org/10.1088/1367-2630/ac631a} {\bibfield  {journal}
  {\bibinfo  {journal} {New J. Phys.}\ }\textbf {\bibinfo {volume} {24}},\
  \bibinfo {pages} {043036} (\bibinfo {year} {2022})}\BibitemShut {NoStop}%
\bibitem [{\citenamefont {Shi}\ \emph {et~al.}(2021)\citenamefont {Shi},
  \citenamefont {Li}, \citenamefont {Chen},\ and\ \citenamefont
  {Zhang}}]{Shi22}%
  \BibitemOpen
  \bibfield  {author} {\bibinfo {author} {\bibfnamefont {F.}~\bibnamefont
  {Shi}}, \bibinfo {author} {\bibfnamefont {M.-S.}\ \bibnamefont {Li}},
  \bibinfo {author} {\bibfnamefont {L.}~\bibnamefont {Chen}},\ and\ \bibinfo
  {author} {\bibfnamefont {X.}~\bibnamefont {Zhang}},\ }\href
  {https://doi.org/10.1088/1751-8121/ac3bea} {\bibfield  {journal} {\bibinfo
  {journal} {J. Phys. A: Math. Theor.}\ }\textbf {\bibinfo {volume} {55}},\
  \bibinfo {pages} {015305} (\bibinfo {year} {2021})}\BibitemShut {NoStop}%
\bibitem [{\citenamefont {Ghosh}\ \emph {et~al.}(2001)\citenamefont {Ghosh},
  \citenamefont {Kar}, \citenamefont {Roy}, \citenamefont {Sen(De)},\ and\
  \citenamefont {Sen}}]{Ghosh01}%
  \BibitemOpen
  \bibfield  {author} {\bibinfo {author} {\bibfnamefont {S.}~\bibnamefont
  {Ghosh}}, \bibinfo {author} {\bibfnamefont {G.}~\bibnamefont {Kar}}, \bibinfo
  {author} {\bibfnamefont {A.}~\bibnamefont {Roy}}, \bibinfo {author}
  {\bibfnamefont {A.}~\bibnamefont {Sen(De)}},\ and\ \bibinfo {author}
  {\bibfnamefont {U.}~\bibnamefont {Sen}},\ }\href
  {https://doi.org/10.1103/PhysRevLett.87.277902} {\bibfield  {journal}
  {\bibinfo  {journal} {Phys. Rev. Lett.}\ }\textbf {\bibinfo {volume} {87}},\
  \bibinfo {pages} {277902} (\bibinfo {year} {2001})}\BibitemShut {NoStop}%
\bibitem [{\citenamefont {Ghosh}\ \emph {et~al.}(2002)\citenamefont {Ghosh},
  \citenamefont {Kar}, \citenamefont {Roy}, \citenamefont {Sarkar},
  \citenamefont {Sen(De)},\ and\ \citenamefont {Sen}}]{Ghosh02}%
  \BibitemOpen
  \bibfield  {author} {\bibinfo {author} {\bibfnamefont {S.}~\bibnamefont
  {Ghosh}}, \bibinfo {author} {\bibfnamefont {G.}~\bibnamefont {Kar}}, \bibinfo
  {author} {\bibfnamefont {A.}~\bibnamefont {Roy}}, \bibinfo {author}
  {\bibfnamefont {D.}~\bibnamefont {Sarkar}}, \bibinfo {author} {\bibfnamefont
  {A.}~\bibnamefont {Sen(De)}},\ and\ \bibinfo {author} {\bibfnamefont
  {U.}~\bibnamefont {Sen}},\ }\href
  {https://doi.org/10.1103/PhysRevA.65.062307} {\bibfield  {journal} {\bibinfo
  {journal} {Phys. Rev. A}\ }\textbf {\bibinfo {volume} {65}},\ \bibinfo
  {pages} {062307} (\bibinfo {year} {2002})}\BibitemShut {NoStop}%
\bibitem [{\citenamefont {Ghosh}\ \emph {et~al.}(2004)\citenamefont {Ghosh},
  \citenamefont {Kar}, \citenamefont {Roy},\ and\ \citenamefont
  {Sarkar}}]{Ghosh04}%
  \BibitemOpen
  \bibfield  {author} {\bibinfo {author} {\bibfnamefont {S.}~\bibnamefont
  {Ghosh}}, \bibinfo {author} {\bibfnamefont {G.}~\bibnamefont {Kar}}, \bibinfo
  {author} {\bibfnamefont {A.}~\bibnamefont {Roy}},\ and\ \bibinfo {author}
  {\bibfnamefont {D.}~\bibnamefont {Sarkar}},\ }\href
  {https://doi.org/10.1103/PhysRevA.70.022304} {\bibfield  {journal} {\bibinfo
  {journal} {Phys. Rev. A}\ }\textbf {\bibinfo {volume} {70}},\ \bibinfo
  {pages} {022304} (\bibinfo {year} {2004})}\BibitemShut {NoStop}%
\bibitem [{\citenamefont {Fan}(2004)}]{Fan04}%
  \BibitemOpen
  \bibfield  {author} {\bibinfo {author} {\bibfnamefont {H.}~\bibnamefont
  {Fan}},\ }\href {https://doi.org/10.1103/PhysRevLett.92.177905} {\bibfield
  {journal} {\bibinfo  {journal} {Phys. Rev. Lett.}\ }\textbf {\bibinfo
  {volume} {92}},\ \bibinfo {pages} {177905} (\bibinfo {year}
  {2004})}\BibitemShut {NoStop}%
\bibitem [{\citenamefont {Nathanson}(2005)}]{Nathanson05}%
  \BibitemOpen
  \bibfield  {author} {\bibinfo {author} {\bibfnamefont {M.}~\bibnamefont
  {Nathanson}},\ }\href {https://doi.org/https://doi.org/10.1063/1.1914731}
  {\bibfield  {journal} {\bibinfo  {journal} {J. Math. Phys.}\ }\textbf
  {\bibinfo {volume} {46}},\ \bibinfo {pages} {062103} (\bibinfo {year}
  {2005})}\BibitemShut {NoStop}%
\bibitem [{\citenamefont {Watrous}(2005)}]{Watrous05}%
  \BibitemOpen
  \bibfield  {author} {\bibinfo {author} {\bibfnamefont {J.}~\bibnamefont
  {Watrous}},\ }\href {https://doi.org/10.1103/PhysRevLett.95.080505}
  {\bibfield  {journal} {\bibinfo  {journal} {Phys. Rev. Lett.}\ }\textbf
  {\bibinfo {volume} {95}},\ \bibinfo {pages} {080505} (\bibinfo {year}
  {2005})}\BibitemShut {NoStop}%
\bibitem [{\citenamefont {Owari}\ and\ \citenamefont
  {Hayashi}(2006)}]{Owari06}%
  \BibitemOpen
  \bibfield  {author} {\bibinfo {author} {\bibfnamefont {M.}~\bibnamefont
  {Owari}}\ and\ \bibinfo {author} {\bibfnamefont {M.}~\bibnamefont
  {Hayashi}},\ }\href {https://doi.org/10.1103/PhysRevA.74.032108} {\bibfield
  {journal} {\bibinfo  {journal} {Phys. Rev. A}\ }\textbf {\bibinfo {volume}
  {74}},\ \bibinfo {pages} {032108} (\bibinfo {year} {2006})}\BibitemShut
  {NoStop}%
\bibitem [{\citenamefont {Xin}\ and\ \citenamefont {Duan}(2008)}]{Xin08}%
  \BibitemOpen
  \bibfield  {author} {\bibinfo {author} {\bibfnamefont {Y.}~\bibnamefont
  {Xin}}\ and\ \bibinfo {author} {\bibfnamefont {R.}~\bibnamefont {Duan}},\
  }\href {https://doi.org/10.1103/PhysRevA.77.012315} {\bibfield  {journal}
  {\bibinfo  {journal} {Phys. Rev. A}\ }\textbf {\bibinfo {volume} {77}},\
  \bibinfo {pages} {012315} (\bibinfo {year} {2008})}\BibitemShut {NoStop}%
\bibitem [{\citenamefont {Bandyopadhyay}(2010)}]{Bandyopadhyay10}%
  \BibitemOpen
  \bibfield  {author} {\bibinfo {author} {\bibfnamefont {S.}~\bibnamefont
  {Bandyopadhyay}},\ }\href {https://doi.org/10.1103/PhysRevA.81.022327}
  {\bibfield  {journal} {\bibinfo  {journal} {Phys. Rev. A}\ }\textbf {\bibinfo
  {volume} {81}},\ \bibinfo {pages} {022327} (\bibinfo {year}
  {2010})}\BibitemShut {NoStop}%
\bibitem [{\citenamefont {Bandyopadhyay}(2011)}]{Bandyopadhyay11}%
  \BibitemOpen
  \bibfield  {author} {\bibinfo {author} {\bibfnamefont {S.}~\bibnamefont
  {Bandyopadhyay}},\ }\href {https://doi.org/10.1103/PhysRevLett.106.210402}
  {\bibfield  {journal} {\bibinfo  {journal} {Phys. Rev. Lett.}\ }\textbf
  {\bibinfo {volume} {106}},\ \bibinfo {pages} {210402} (\bibinfo {year}
  {2011})}\BibitemShut {NoStop}%
\bibitem [{\citenamefont {Yu}\ \emph {et~al.}(2012)\citenamefont {Yu},
  \citenamefont {Duan},\ and\ \citenamefont {Ying}}]{Yu12}%
  \BibitemOpen
  \bibfield  {author} {\bibinfo {author} {\bibfnamefont {N.}~\bibnamefont
  {Yu}}, \bibinfo {author} {\bibfnamefont {R.}~\bibnamefont {Duan}},\ and\
  \bibinfo {author} {\bibfnamefont {M.}~\bibnamefont {Ying}},\ }\href
  {https://doi.org/10.1103/PhysRevLett.109.020506} {\bibfield  {journal}
  {\bibinfo  {journal} {Phys. Rev. Lett.}\ }\textbf {\bibinfo {volume} {109}},\
  \bibinfo {pages} {020506} (\bibinfo {year} {2012})}\BibitemShut {NoStop}%
\bibitem [{\citenamefont {Yang}\ \emph {et~al.}(2013)\citenamefont {Yang},
  \citenamefont {Gao}, \citenamefont {Tian}, \citenamefont {Cao},\ and\
  \citenamefont {Wen}}]{Yang13}%
  \BibitemOpen
  \bibfield  {author} {\bibinfo {author} {\bibfnamefont {Y.-H.}\ \bibnamefont
  {Yang}}, \bibinfo {author} {\bibfnamefont {F.}~\bibnamefont {Gao}}, \bibinfo
  {author} {\bibfnamefont {G.-J.}\ \bibnamefont {Tian}}, \bibinfo {author}
  {\bibfnamefont {T.-Q.}\ \bibnamefont {Cao}},\ and\ \bibinfo {author}
  {\bibfnamefont {Q.-Y.}\ \bibnamefont {Wen}},\ }\href
  {https://doi.org/10.1103/PhysRevA.88.024301} {\bibfield  {journal} {\bibinfo
  {journal} {Phys. Rev. A}\ }\textbf {\bibinfo {volume} {88}},\ \bibinfo
  {pages} {024301} (\bibinfo {year} {2013})}\BibitemShut {NoStop}%
\bibitem [{\citenamefont {Cosentino}\ and\ \citenamefont
  {Russo}(2014)}]{Cosentino14}%
  \BibitemOpen
  \bibfield  {author} {\bibinfo {author} {\bibfnamefont {A.}~\bibnamefont
  {Cosentino}}\ and\ \bibinfo {author} {\bibfnamefont {V.}~\bibnamefont
  {Russo}},\ }\href {https://doi.org/10.48550/arXiv.1307.3232} {\bibfield
  {journal} {\bibinfo  {journal} {Quantum Inf. Comput.}\ }\textbf {\bibinfo
  {volume} {14}},\ \bibinfo {pages} {1098} (\bibinfo {year}
  {2014})}\BibitemShut {NoStop}%
\bibitem [{\citenamefont {Singal}(2016)}]{Singal16}%
  \BibitemOpen
  \bibfield  {author} {\bibinfo {author} {\bibfnamefont {T.}~\bibnamefont
  {Singal}},\ }\href {https://doi.org/10.1103/PhysRevA.93.030301} {\bibfield
  {journal} {\bibinfo  {journal} {Phys. Rev. A}\ }\textbf {\bibinfo {volume}
  {93}},\ \bibinfo {pages} {030301(R)} (\bibinfo {year} {2016})}\BibitemShut
  {NoStop}%
\bibitem [{\citenamefont {Singal}\ \emph {et~al.}(2017)\citenamefont {Singal},
  \citenamefont {Rahaman}, \citenamefont {Ghosh},\ and\ \citenamefont
  {Kar}}]{Singal17}%
  \BibitemOpen
  \bibfield  {author} {\bibinfo {author} {\bibfnamefont {T.}~\bibnamefont
  {Singal}}, \bibinfo {author} {\bibfnamefont {R.}~\bibnamefont {Rahaman}},
  \bibinfo {author} {\bibfnamefont {S.}~\bibnamefont {Ghosh}},\ and\ \bibinfo
  {author} {\bibfnamefont {G.}~\bibnamefont {Kar}},\ }\href
  {https://doi.org/10.1103/PhysRevA.96.042314} {\bibfield  {journal} {\bibinfo
  {journal} {Phys. Rev. A}\ }\textbf {\bibinfo {volume} {96}},\ \bibinfo
  {pages} {042314} (\bibinfo {year} {2017})}\BibitemShut {NoStop}%
\bibitem [{\citenamefont {Yuan}\ \emph {et~al.}(2020)\citenamefont {Yuan},
  \citenamefont {Yang},\ and\ \citenamefont {Wang}}]{Yuan20}%
  \BibitemOpen
  \bibfield  {author} {\bibinfo {author} {\bibfnamefont {J.-T.}\ \bibnamefont
  {Yuan}}, \bibinfo {author} {\bibfnamefont {Y.-H.}\ \bibnamefont {Yang}},\
  and\ \bibinfo {author} {\bibfnamefont {C.-H.}\ \bibnamefont {Wang}},\ }\href
  {https://doi.org/10.1088/1751-8121/abc43b} {\bibfield  {journal} {\bibinfo
  {journal} {J. Phys. A: Math. Theor.}\ }\textbf {\bibinfo {volume} {53}},\
  \bibinfo {pages} {505304} (\bibinfo {year} {2020})}\BibitemShut {NoStop}%
\bibitem [{\citenamefont {Banik}\ \emph {et~al.}(2021)\citenamefont {Banik},
  \citenamefont {Guha}, \citenamefont {Alimuddin}, \citenamefont {Kar},
  \citenamefont {Halder},\ and\ \citenamefont {Bhattacharya}}]{Banik21}%
  \BibitemOpen
  \bibfield  {author} {\bibinfo {author} {\bibfnamefont {M.}~\bibnamefont
  {Banik}}, \bibinfo {author} {\bibfnamefont {T.}~\bibnamefont {Guha}},
  \bibinfo {author} {\bibfnamefont {M.}~\bibnamefont {Alimuddin}}, \bibinfo
  {author} {\bibfnamefont {G.}~\bibnamefont {Kar}}, \bibinfo {author}
  {\bibfnamefont {S.}~\bibnamefont {Halder}},\ and\ \bibinfo {author}
  {\bibfnamefont {S.~S.}\ \bibnamefont {Bhattacharya}},\ }\href
  {https://doi.org/10.1103/PhysRevLett.126.210505} {\bibfield  {journal}
  {\bibinfo  {journal} {Phys. Rev. Lett.}\ }\textbf {\bibinfo {volume} {126}},\
  \bibinfo {pages} {210505} (\bibinfo {year} {2021})}\BibitemShut {NoStop}%
\bibitem [{\citenamefont {Ha}\ and\ \citenamefont {Kim}(2023)}]{Ha23}%
  \BibitemOpen
  \bibfield  {author} {\bibinfo {author} {\bibfnamefont {D.}~\bibnamefont
  {Ha}}\ and\ \bibinfo {author} {\bibfnamefont {J.~S.}\ \bibnamefont {Kim}},\
  }\href {https://doi.org/10.1103/PhysRevA.107.052410} {\bibfield  {journal}
  {\bibinfo  {journal} {Phys. Rev. A}\ }\textbf {\bibinfo {volume} {107}},\
  \bibinfo {pages} {052410} (\bibinfo {year} {2023})}\BibitemShut {NoStop}%
\bibitem [{\citenamefont {Liu}\ \emph {et~al.}(2023)\citenamefont {Liu},
  \citenamefont {Yang}, \citenamefont {Chen},\ and\ \citenamefont
  {Wang}}]{Liu23}%
  \BibitemOpen
  \bibfield  {author} {\bibinfo {author} {\bibfnamefont {Q.-Q.}\ \bibnamefont
  {Liu}}, \bibinfo {author} {\bibfnamefont {Y.-H.}\ \bibnamefont {Yang}},
  \bibinfo {author} {\bibfnamefont {P.-Y.}\ \bibnamefont {Chen}},\ and\
  \bibinfo {author} {\bibfnamefont {X.-L.}\ \bibnamefont {Wang}},\ }\href
  {https://doi.org/10.1088/1402-4896/acfbc1} {\bibfield  {journal} {\bibinfo
  {journal} {Phys. Scr.}\ }\textbf {\bibinfo {volume} {98}},\ \bibinfo {pages}
  {115102} (\bibinfo {year} {2023})}\BibitemShut {NoStop}%
\bibitem [{\citenamefont {Badzia\ifmmode~\mbox{\c{}}\else \c{}\fi{}g}\ \emph
  {et~al.}(2003)\citenamefont {Badzia\ifmmode~\mbox{\c{}}\else \c{}\fi{}g},
  \citenamefont {Horodecki}, \citenamefont {Sen(De)},\ and\ \citenamefont
  {Sen}}]{Badziag03}%
  \BibitemOpen
  \bibfield  {author} {\bibinfo {author} {\bibfnamefont {P.}~\bibnamefont
  {Badzia\ifmmode~\mbox{\c{}}\else \c{}\fi{}g}}, \bibinfo {author}
  {\bibfnamefont {M.}~\bibnamefont {Horodecki}}, \bibinfo {author}
  {\bibfnamefont {A.}~\bibnamefont {Sen(De)}},\ and\ \bibinfo {author}
  {\bibfnamefont {U.}~\bibnamefont {Sen}},\ }\href
  {https://doi.org/10.1103/PhysRevLett.91.117901} {\bibfield  {journal}
  {\bibinfo  {journal} {Phys. Rev. Lett.}\ }\textbf {\bibinfo {volume} {91}},\
  \bibinfo {pages} {117901} (\bibinfo {year} {2003})}\BibitemShut {NoStop}%
\bibitem [{\citenamefont {Horodecki}\ \emph {et~al.}(2004)\citenamefont
  {Horodecki}, \citenamefont {Oppenheim}, \citenamefont {Sen(De)},\ and\
  \citenamefont {Sen}}]{Horodecki04}%
  \BibitemOpen
  \bibfield  {author} {\bibinfo {author} {\bibfnamefont {M.}~\bibnamefont
  {Horodecki}}, \bibinfo {author} {\bibfnamefont {J.}~\bibnamefont
  {Oppenheim}}, \bibinfo {author} {\bibfnamefont {A.}~\bibnamefont {Sen(De)}},\
  and\ \bibinfo {author} {\bibfnamefont {U.}~\bibnamefont {Sen}},\ }\href
  {https://doi.org/10.1103/PhysRevLett.93.170503} {\bibfield  {journal}
  {\bibinfo  {journal} {Phys. Rev. Lett.}\ }\textbf {\bibinfo {volume} {93}},\
  \bibinfo {pages} {170503} (\bibinfo {year} {2004})}\BibitemShut {NoStop}%
\bibitem [{\citenamefont {Rahaman}\ and\ \citenamefont
  {Parker}(2015)}]{Rahaman15}%
  \BibitemOpen
  \bibfield  {author} {\bibinfo {author} {\bibfnamefont {R.}~\bibnamefont
  {Rahaman}}\ and\ \bibinfo {author} {\bibfnamefont {M.~G.}\ \bibnamefont
  {Parker}},\ }\href {https://doi.org/10.1103/PhysRevA.91.022330} {\bibfield
  {journal} {\bibinfo  {journal} {Phys. Rev. A}\ }\textbf {\bibinfo {volume}
  {91}},\ \bibinfo {pages} {022330} (\bibinfo {year} {2015})}\BibitemShut
  {NoStop}%
\bibitem [{\citenamefont {Lami}\ \emph {et~al.}(2018)\citenamefont {Lami},
  \citenamefont {Palazuelos},\ and\ \citenamefont {Winter}}]{Lami18}%
  \BibitemOpen
  \bibfield  {author} {\bibinfo {author} {\bibfnamefont {L.}~\bibnamefont
  {Lami}}, \bibinfo {author} {\bibfnamefont {C.}~\bibnamefont {Palazuelos}},\
  and\ \bibinfo {author} {\bibfnamefont {A.}~\bibnamefont {Winter}},\ }\href
  {https://doi.org/10.1007/s00220-018-3154-4} {\bibfield  {journal} {\bibinfo
  {journal} {Commun. Math. Phys.}\ }\textbf {\bibinfo {volume} {361}},\
  \bibinfo {pages} {661} (\bibinfo {year} {2018})}\BibitemShut {NoStop}%
\bibitem [{\citenamefont {Lami}(2021)}]{Lami21}%
  \BibitemOpen
  \bibfield  {author} {\bibinfo {author} {\bibfnamefont {L.}~\bibnamefont
  {Lami}},\ }\href {https://doi.org/10.1103/PhysRevA.104.052428} {\bibfield
  {journal} {\bibinfo  {journal} {Phys. Rev. A}\ }\textbf {\bibinfo {volume}
  {104}},\ \bibinfo {pages} {052428} (\bibinfo {year} {2021})}\BibitemShut
  {NoStop}%
\bibitem [{\citenamefont {Ha}\ and\ \citenamefont {Kim}(2025)}]{Ha25}%
  \BibitemOpen
  \bibfield  {author} {\bibinfo {author} {\bibfnamefont {D.}~\bibnamefont
  {Ha}}\ and\ \bibinfo {author} {\bibfnamefont {J.~S.}\ \bibnamefont {Kim}},\
  }\href {https://doi.org/10.1103/PhysRevA.111.052405} {\bibfield  {journal}
  {\bibinfo  {journal} {Phys. Rev. A}\ }\textbf {\bibinfo {volume} {111}},\
  \bibinfo {pages} {052405} (\bibinfo {year} {2025})}\BibitemShut {NoStop}%
\bibitem [{\citenamefont {Horodecki}\ \emph {et~al.}(2009)\citenamefont
  {Horodecki}, \citenamefont {Horodecki}, \citenamefont {Horodecki},\ and\
  \citenamefont {Horodecki}}]{Horodecki09}%
  \BibitemOpen
  \bibfield  {author} {\bibinfo {author} {\bibfnamefont {R.}~\bibnamefont
  {Horodecki}}, \bibinfo {author} {\bibfnamefont {P.}~\bibnamefont
  {Horodecki}}, \bibinfo {author} {\bibfnamefont {M.}~\bibnamefont
  {Horodecki}},\ and\ \bibinfo {author} {\bibfnamefont {K.}~\bibnamefont
  {Horodecki}},\ }\href {https://doi.org/10.1103/RevModPhys.81.865} {\bibfield
  {journal} {\bibinfo  {journal} {Rev. Mod. Phys.}\ }\textbf {\bibinfo {volume}
  {81}},\ \bibinfo {pages} {865} (\bibinfo {year} {2009})}\BibitemShut
  {NoStop}%
\bibitem [{\citenamefont {Cohen}(2008)}]{Cohen08}%
  \BibitemOpen
  \bibfield  {author} {\bibinfo {author} {\bibfnamefont {S.~M.}\ \bibnamefont
  {Cohen}},\ }\href {https://doi.org/10.1103/PhysRevA.77.012304} {\bibfield
  {journal} {\bibinfo  {journal} {Phys. Rev. A}\ }\textbf {\bibinfo {volume}
  {77}},\ \bibinfo {pages} {012304} (\bibinfo {year} {2008})}\BibitemShut
  {NoStop}%
\bibitem [{\citenamefont {Bandyopadhyay}\ \emph {et~al.}(2009)\citenamefont
  {Bandyopadhyay}, \citenamefont {Brassard}, \citenamefont {Kimmel},\ and\
  \citenamefont {Wootters}}]{Bandyopadhyay09}%
  \BibitemOpen
  \bibfield  {author} {\bibinfo {author} {\bibfnamefont {S.}~\bibnamefont
  {Bandyopadhyay}}, \bibinfo {author} {\bibfnamefont {G.}~\bibnamefont
  {Brassard}}, \bibinfo {author} {\bibfnamefont {S.}~\bibnamefont {Kimmel}},\
  and\ \bibinfo {author} {\bibfnamefont {W.~K.}\ \bibnamefont {Wootters}},\
  }\href {https://doi.org/10.1103/PhysRevA.80.012313} {\bibfield  {journal}
  {\bibinfo  {journal} {Phys. Rev. A}\ }\textbf {\bibinfo {volume} {80}},\
  \bibinfo {pages} {012313} (\bibinfo {year} {2009})}\BibitemShut {NoStop}%
\bibitem [{\citenamefont {Bandyopadhyay}\ \emph {et~al.}(2016)\citenamefont
  {Bandyopadhyay}, \citenamefont {Halder},\ and\ \citenamefont
  {Nathanson}}]{Bandyopadhyay16}%
  \BibitemOpen
  \bibfield  {author} {\bibinfo {author} {\bibfnamefont {S.}~\bibnamefont
  {Bandyopadhyay}}, \bibinfo {author} {\bibfnamefont {S.}~\bibnamefont
  {Halder}},\ and\ \bibinfo {author} {\bibfnamefont {M.}~\bibnamefont
  {Nathanson}},\ }\href {https://doi.org/10.1103/PhysRevA.94.022311} {\bibfield
   {journal} {\bibinfo  {journal} {Phys. Rev. A}\ }\textbf {\bibinfo {volume}
  {94}},\ \bibinfo {pages} {022311} (\bibinfo {year} {2016})}\BibitemShut
  {NoStop}%
\bibitem [{\citenamefont {Zhang}\ \emph {et~al.}(2016)\citenamefont {Zhang},
  \citenamefont {Gao}, \citenamefont {Cao}, \citenamefont {Qin},\ and\
  \citenamefont {Wen}}]{Zhang16}%
  \BibitemOpen
  \bibfield  {author} {\bibinfo {author} {\bibfnamefont {Z.-C.}\ \bibnamefont
  {Zhang}}, \bibinfo {author} {\bibfnamefont {F.}~\bibnamefont {Gao}}, \bibinfo
  {author} {\bibfnamefont {T.-Q.}\ \bibnamefont {Cao}}, \bibinfo {author}
  {\bibfnamefont {S.-J.}\ \bibnamefont {Qin}},\ and\ \bibinfo {author}
  {\bibfnamefont {Q.-Y.}\ \bibnamefont {Wen}},\ }\href
  {https://doi.org/10.1038/srep30493} {\bibfield  {journal} {\bibinfo
  {journal} {Sci. Rep.}\ }\textbf {\bibinfo {volume} {6}},\ \bibinfo {pages}
  {30493} (\bibinfo {year} {2016})}\BibitemShut {NoStop}%
\bibitem [{\citenamefont {Zhang}\ \emph {et~al.}(2018)\citenamefont {Zhang},
  \citenamefont {Song}, \citenamefont {Song}, \citenamefont {Gao},
  \citenamefont {Qin},\ and\ \citenamefont {Wen}}]{Zhang18}%
  \BibitemOpen
  \bibfield  {author} {\bibinfo {author} {\bibfnamefont {Z.-C.}\ \bibnamefont
  {Zhang}}, \bibinfo {author} {\bibfnamefont {Y.-Q.}\ \bibnamefont {Song}},
  \bibinfo {author} {\bibfnamefont {T.-T.}\ \bibnamefont {Song}}, \bibinfo
  {author} {\bibfnamefont {F.}~\bibnamefont {Gao}}, \bibinfo {author}
  {\bibfnamefont {S.-J.}\ \bibnamefont {Qin}},\ and\ \bibinfo {author}
  {\bibfnamefont {Q.-Y.}\ \bibnamefont {Wen}},\ }\href
  {https://doi.org/10.1103/PhysRevA.97.022334} {\bibfield  {journal} {\bibinfo
  {journal} {Phys. Rev. A}\ }\textbf {\bibinfo {volume} {97}},\ \bibinfo
  {pages} {022334} (\bibinfo {year} {2018})}\BibitemShut {NoStop}%
\bibitem [{\citenamefont {Rout}\ \emph {et~al.}(2019)\citenamefont {Rout},
  \citenamefont {Maity}, \citenamefont {Mukherjee}, \citenamefont {Halder},\
  and\ \citenamefont {Banik}}]{Rout19}%
  \BibitemOpen
  \bibfield  {author} {\bibinfo {author} {\bibfnamefont {S.}~\bibnamefont
  {Rout}}, \bibinfo {author} {\bibfnamefont {A.~G.}\ \bibnamefont {Maity}},
  \bibinfo {author} {\bibfnamefont {A.}~\bibnamefont {Mukherjee}}, \bibinfo
  {author} {\bibfnamefont {S.}~\bibnamefont {Halder}},\ and\ \bibinfo {author}
  {\bibfnamefont {M.}~\bibnamefont {Banik}},\ }\href
  {https://doi.org/10.1103/PhysRevA.100.032321} {\bibfield  {journal} {\bibinfo
   {journal} {Phys. Rev. A}\ }\textbf {\bibinfo {volume} {100}},\ \bibinfo
  {pages} {032321} (\bibinfo {year} {2019})}\BibitemShut {NoStop}%
\bibitem [{\citenamefont {Li}\ \emph {et~al.}(2019{\natexlab{a}})\citenamefont
  {Li}, \citenamefont {Gao}, \citenamefont {Zhang},\ and\ \citenamefont
  {Wen}}]{Li19}%
  \BibitemOpen
  \bibfield  {author} {\bibinfo {author} {\bibfnamefont {L.-J.}\ \bibnamefont
  {Li}}, \bibinfo {author} {\bibfnamefont {F.}~\bibnamefont {Gao}}, \bibinfo
  {author} {\bibfnamefont {Z.-C.}\ \bibnamefont {Zhang}},\ and\ \bibinfo
  {author} {\bibfnamefont {Q.-Y.}\ \bibnamefont {Wen}},\ }\href
  {https://doi.org/10.1103/PhysRevA.99.012343} {\bibfield  {journal} {\bibinfo
  {journal} {Phys. Rev. A}\ }\textbf {\bibinfo {volume} {99}},\ \bibinfo
  {pages} {012343} (\bibinfo {year} {2019}{\natexlab{a}})}\BibitemShut
  {NoStop}%
\bibitem [{\citenamefont {Li}\ \emph {et~al.}(2019{\natexlab{b}})\citenamefont
  {Li}, \citenamefont {Gao}, \citenamefont {Zhang},\ and\ \citenamefont
  {Wen}}]{Li19-1}%
  \BibitemOpen
  \bibfield  {author} {\bibinfo {author} {\bibfnamefont {L.-J.}\ \bibnamefont
  {Li}}, \bibinfo {author} {\bibfnamefont {F.}~\bibnamefont {Gao}}, \bibinfo
  {author} {\bibfnamefont {Z.-C.}\ \bibnamefont {Zhang}},\ and\ \bibinfo
  {author} {\bibfnamefont {Q.-Y.}\ \bibnamefont {Wen}},\ }\href
  {https://doi.org/10.1007/s11128-019-2441-0} {\bibfield  {journal} {\bibinfo
  {journal} {Quantum Inf. Process.}\ }\textbf {\bibinfo {volume} {18}},\
  \bibinfo {pages} {330} (\bibinfo {year} {2019}{\natexlab{b}})}\BibitemShut
  {NoStop}%
\bibitem [{\citenamefont {Shi}\ \emph {et~al.}(2020{\natexlab{a}})\citenamefont
  {Shi}, \citenamefont {Zhang},\ and\ \citenamefont {Chen}}]{Shi20}%
  \BibitemOpen
  \bibfield  {author} {\bibinfo {author} {\bibfnamefont {F.}~\bibnamefont
  {Shi}}, \bibinfo {author} {\bibfnamefont {X.}~\bibnamefont {Zhang}},\ and\
  \bibinfo {author} {\bibfnamefont {L.}~\bibnamefont {Chen}},\ }\href
  {https://doi.org/10.1103/PhysRevA.101.062329} {\bibfield  {journal} {\bibinfo
   {journal} {Phys. Rev. A}\ }\textbf {\bibinfo {volume} {101}},\ \bibinfo
  {pages} {062329} (\bibinfo {year} {2020}{\natexlab{a}})}\BibitemShut
  {NoStop}%
\bibitem [{\citenamefont {Zhang}\ \emph {et~al.}(2020)\citenamefont {Zhang},
  \citenamefont {Wu},\ and\ \citenamefont {Zhang}}]{Zhang20}%
  \BibitemOpen
  \bibfield  {author} {\bibinfo {author} {\bibfnamefont {Z.-C.}\ \bibnamefont
  {Zhang}}, \bibinfo {author} {\bibfnamefont {X.}~\bibnamefont {Wu}},\ and\
  \bibinfo {author} {\bibfnamefont {X.}~\bibnamefont {Zhang}},\ }\href
  {https://doi.org/10.1103/PhysRevA.101.022306} {\bibfield  {journal} {\bibinfo
   {journal} {Phys. Rev. A}\ }\textbf {\bibinfo {volume} {101}},\ \bibinfo
  {pages} {022306} (\bibinfo {year} {2020})}\BibitemShut {NoStop}%
\bibitem [{\citenamefont {Zhang}\ and\ \citenamefont {Wang}(2021)}]{Zhang21}%
  \BibitemOpen
  \bibfield  {author} {\bibinfo {author} {\bibfnamefont {Z.-C.}\ \bibnamefont
  {Zhang}}\ and\ \bibinfo {author} {\bibfnamefont {Q.-L.}\ \bibnamefont
  {Wang}},\ }\href {https://doi.org/10.1007/s11128-021-03016-0} {\bibfield
  {journal} {\bibinfo  {journal} {Quantum Inf. Process.}\ }\textbf {\bibinfo
  {volume} {20}},\ \bibinfo {pages} {75} (\bibinfo {year} {2021})}\BibitemShut
  {NoStop}%
\bibitem [{\citenamefont {Cao}\ \emph {et~al.}(2021)\citenamefont {Cao},
  \citenamefont {Xin},\ and\ \citenamefont {Zhang}}]{Cao21}%
  \BibitemOpen
  \bibfield  {author} {\bibinfo {author} {\bibfnamefont {T.-Q.}\ \bibnamefont
  {Cao}}, \bibinfo {author} {\bibfnamefont {Q.-L.}\ \bibnamefont {Xin}},\ and\
  \bibinfo {author} {\bibfnamefont {Z.-C.}\ \bibnamefont {Zhang}},\ }\href
  {https://doi.org/10.1007/s11128-021-03313-8} {\bibfield  {journal} {\bibinfo
  {journal} {Quantum Inf. Process.}\ }\textbf {\bibinfo {volume} {20}},\
  \bibinfo {pages} {362} (\bibinfo {year} {2021})}\BibitemShut {NoStop}%
\bibitem [{\citenamefont {Zhang}\ \emph {et~al.}(2022)\citenamefont {Zhang},
  \citenamefont {Wei},\ and\ \citenamefont {Wang}}]{Zhang22}%
  \BibitemOpen
  \bibfield  {author} {\bibinfo {author} {\bibfnamefont {Z.-C.}\ \bibnamefont
  {Zhang}}, \bibinfo {author} {\bibfnamefont {X.-J.}\ \bibnamefont {Wei}},\
  and\ \bibinfo {author} {\bibfnamefont {A.-L.}\ \bibnamefont {Wang}},\ }\href
  {https://doi.org/10.1007/s11128-022-03696-2} {\bibfield  {journal} {\bibinfo
  {journal} {Quantum Inf. Process.}\ }\textbf {\bibinfo {volume} {21}},\
  \bibinfo {pages} {342} (\bibinfo {year} {2022})}\BibitemShut {NoStop}%
\bibitem [{\citenamefont {Zhou}\ \emph {et~al.}(2023)\citenamefont {Zhou},
  \citenamefont {Gao},\ and\ \citenamefont {Yan}}]{Zhou23}%
  \BibitemOpen
  \bibfield  {author} {\bibinfo {author} {\bibfnamefont {H.}~\bibnamefont
  {Zhou}}, \bibinfo {author} {\bibfnamefont {T.}~\bibnamefont {Gao}},\ and\
  \bibinfo {author} {\bibfnamefont {F.}~\bibnamefont {Yan}},\ }\href
  {https://doi.org/10.1103/PhysRevA.107.042214} {\bibfield  {journal} {\bibinfo
   {journal} {Phys. Rev. A}\ }\textbf {\bibinfo {volume} {107}},\ \bibinfo
  {pages} {042214} (\bibinfo {year} {2023})}\BibitemShut {NoStop}%
\bibitem [{\citenamefont {Bhunia}\ \emph {et~al.}(2023)\citenamefont {Bhunia},
  \citenamefont {Biswas}, \citenamefont {Chattopadhyay},\ and\ \citenamefont
  {Sarkar}}]{Bhunia23}%
  \BibitemOpen
  \bibfield  {author} {\bibinfo {author} {\bibfnamefont {A.}~\bibnamefont
  {Bhunia}}, \bibinfo {author} {\bibfnamefont {I.}~\bibnamefont {Biswas}},
  \bibinfo {author} {\bibfnamefont {I.}~\bibnamefont {Chattopadhyay}},\ and\
  \bibinfo {author} {\bibfnamefont {D.}~\bibnamefont {Sarkar}},\ }\href
  {https://doi.org/10.1088/1751-8121/aceddb} {\bibfield  {journal} {\bibinfo
  {journal} {J. Phys. A: Math. Theor.}\ }\textbf {\bibinfo {volume} {56}},\
  \bibinfo {pages} {365303} (\bibinfo {year} {2023})}\BibitemShut {NoStop}%
\bibitem [{\citenamefont {Cao}\ and\ \citenamefont {Zuo}(2023)}]{Cao23}%
  \BibitemOpen
  \bibfield  {author} {\bibinfo {author} {\bibfnamefont {H.-Q.}\ \bibnamefont
  {Cao}}\ and\ \bibinfo {author} {\bibfnamefont {H.-J.}\ \bibnamefont {Zuo}},\
  }\href {https://doi.org/https://doi.org/10.1016/j.physa.2023.128852}
  {\bibfield  {journal} {\bibinfo  {journal} {Phys. A: Stat. Mech. Appl.}\
  }\textbf {\bibinfo {volume} {623}},\ \bibinfo {pages} {128852} (\bibinfo
  {year} {2023})}\BibitemShut {NoStop}%
\bibitem [{\citenamefont {Qiao}\ \emph {et~al.}(2024)\citenamefont {Qiao},
  \citenamefont {Zhang}, \citenamefont {Bai},\ and\ \citenamefont
  {Liu}}]{Qiao24}%
  \BibitemOpen
  \bibfield  {author} {\bibinfo {author} {\bibfnamefont {Q.}~\bibnamefont
  {Qiao}}, \bibinfo {author} {\bibfnamefont {S.-J.}\ \bibnamefont {Zhang}},
  \bibinfo {author} {\bibfnamefont {C.-M.}\ \bibnamefont {Bai}},\ and\ \bibinfo
  {author} {\bibfnamefont {L.}~\bibnamefont {Liu}},\ }\href
  {https://doi.org/10.1088/1572-9494/ad6de6} {\bibfield  {journal} {\bibinfo
  {journal} {Commun. Theor. Phys.}\ }\textbf {\bibinfo {volume} {76}},\
  \bibinfo {pages} {125101} (\bibinfo {year} {2024})}\BibitemShut {NoStop}%
\bibitem [{\citenamefont {Wei}\ \emph {et~al.}(2024)\citenamefont {Wei},
  \citenamefont {Xie}, \citenamefont {Li},\ and\ \citenamefont
  {Zhang}}]{Wei24}%
  \BibitemOpen
  \bibfield  {author} {\bibinfo {author} {\bibfnamefont {X.-J.}\ \bibnamefont
  {Wei}}, \bibinfo {author} {\bibfnamefont {Z.-S.}\ \bibnamefont {Xie}},
  \bibinfo {author} {\bibfnamefont {Y.-L.}\ \bibnamefont {Li}},\ and\ \bibinfo
  {author} {\bibfnamefont {Z.-C.}\ \bibnamefont {Zhang}},\ }\href
  {https://doi.org/10.1007/s11128-024-04558-9} {\bibfield  {journal} {\bibinfo
  {journal} {Quantum Inf. Process.}\ }\textbf {\bibinfo {volume} {23}},\
  \bibinfo {pages} {361} (\bibinfo {year} {2024})}\BibitemShut {NoStop}%
\bibitem [{\citenamefont {Zhang}\ \emph {et~al.}(2025)\citenamefont {Zhang},
  \citenamefont {Qiao},\ and\ \citenamefont {Bai}}]{Zhang25}%
  \BibitemOpen
  \bibfield  {author} {\bibinfo {author} {\bibfnamefont {S.-J.}\ \bibnamefont
  {Zhang}}, \bibinfo {author} {\bibfnamefont {Q.}~\bibnamefont {Qiao}},\ and\
  \bibinfo {author} {\bibfnamefont {C.-M.}\ \bibnamefont {Bai}},\ }\href
  {https://doi.org/10.1007/s11128-025-04718-5} {\bibfield  {journal} {\bibinfo
  {journal} {Quantum Inf. Process.}\ }\textbf {\bibinfo {volume} {24}},\
  \bibinfo {pages} {101} (\bibinfo {year} {2025})}\BibitemShut {NoStop}%
\bibitem [{\citenamefont {Cao}\ \emph {et~al.}(2025)\citenamefont {Cao},
  \citenamefont {Gao}, \citenamefont {Xin},\ and\ \citenamefont
  {Zhao}}]{Cao25}%
  \BibitemOpen
  \bibfield  {author} {\bibinfo {author} {\bibfnamefont {T.-Q.}\ \bibnamefont
  {Cao}}, \bibinfo {author} {\bibfnamefont {B.-H.}\ \bibnamefont {Gao}},
  \bibinfo {author} {\bibfnamefont {Q.-L.}\ \bibnamefont {Xin}},\ and\ \bibinfo
  {author} {\bibfnamefont {L.}~\bibnamefont {Zhao}},\ }\href
  {https://doi.org/10.1007/s11128-025-04727-4} {\bibfield  {journal} {\bibinfo
  {journal} {Quantum Inf. Process.}\ }\textbf {\bibinfo {volume} {24}},\
  \bibinfo {pages} {104} (\bibinfo {year} {2025})}\BibitemShut {NoStop}%
\bibitem [{\citenamefont {Yu}\ \emph {et~al.}(2014)\citenamefont {Yu},
  \citenamefont {Duan},\ and\ \citenamefont {Ying}}]{Yu14}%
  \BibitemOpen
  \bibfield  {author} {\bibinfo {author} {\bibfnamefont {N.}~\bibnamefont
  {Yu}}, \bibinfo {author} {\bibfnamefont {R.}~\bibnamefont {Duan}},\ and\
  \bibinfo {author} {\bibfnamefont {M.}~\bibnamefont {Ying}},\ }\href
  {https://doi.org/10.1109/TIT.2014.2307575} {\bibfield  {journal} {\bibinfo
  {journal} {IEEE Trans. Inf. Theory}\ }\textbf {\bibinfo {volume} {60}},\
  \bibinfo {pages} {2069} (\bibinfo {year} {2014})}\BibitemShut {NoStop}%
\bibitem [{\citenamefont {G\"ung\"or}\ and\ \citenamefont
  {Turgut}(2016)}]{Gungor16}%
  \BibitemOpen
  \bibfield  {author} {\bibinfo {author} {\bibfnamefont {O.~m.~c.}\
  \bibnamefont {G\"ung\"or}}\ and\ \bibinfo {author} {\bibfnamefont
  {S.}~\bibnamefont {Turgut}},\ }\href
  {https://doi.org/10.1103/PhysRevA.94.032330} {\bibfield  {journal} {\bibinfo
  {journal} {Phys. Rev. A}\ }\textbf {\bibinfo {volume} {94}},\ \bibinfo
  {pages} {032330} (\bibinfo {year} {2016})}\BibitemShut {NoStop}%
\bibitem [{\citenamefont {Bandyopadhyay}\ \emph {et~al.}(2018)\citenamefont
  {Bandyopadhyay}, \citenamefont {Halder},\ and\ \citenamefont
  {Nathanson}}]{Bandyopadhyay18}%
  \BibitemOpen
  \bibfield  {author} {\bibinfo {author} {\bibfnamefont {S.}~\bibnamefont
  {Bandyopadhyay}}, \bibinfo {author} {\bibfnamefont {S.}~\bibnamefont
  {Halder}},\ and\ \bibinfo {author} {\bibfnamefont {M.}~\bibnamefont
  {Nathanson}},\ }\href {https://doi.org/10.1103/PhysRevA.97.022314} {\bibfield
   {journal} {\bibinfo  {journal} {Phys. Rev. A}\ }\textbf {\bibinfo {volume}
  {97}},\ \bibinfo {pages} {022314} (\bibinfo {year} {2018})}\BibitemShut
  {NoStop}%
\bibitem [{\citenamefont {Shi}\ \emph {et~al.}(2020{\natexlab{b}})\citenamefont
  {Shi}, \citenamefont {Hu}, \citenamefont {Chen},\ and\ \citenamefont
  {Zhang}}]{Shi20-1}%
  \BibitemOpen
  \bibfield  {author} {\bibinfo {author} {\bibfnamefont {F.}~\bibnamefont
  {Shi}}, \bibinfo {author} {\bibfnamefont {M.}~\bibnamefont {Hu}}, \bibinfo
  {author} {\bibfnamefont {L.}~\bibnamefont {Chen}},\ and\ \bibinfo {author}
  {\bibfnamefont {X.}~\bibnamefont {Zhang}},\ }\href
  {https://doi.org/10.1103/PhysRevA.102.042202} {\bibfield  {journal} {\bibinfo
   {journal} {Phys. Rev. A}\ }\textbf {\bibinfo {volume} {102}},\ \bibinfo
  {pages} {042202} (\bibinfo {year} {2020}{\natexlab{b}})}\BibitemShut
  {NoStop}%
\bibitem [{\citenamefont {Yang}\ \emph {et~al.}(2020)\citenamefont {Yang},
  \citenamefont {Yuan}, \citenamefont {Wang},\ and\ \citenamefont
  {Gao}}]{Yang20}%
  \BibitemOpen
  \bibfield  {author} {\bibinfo {author} {\bibfnamefont {Y.-H.}\ \bibnamefont
  {Yang}}, \bibinfo {author} {\bibfnamefont {J.-T.}\ \bibnamefont {Yuan}},
  \bibinfo {author} {\bibfnamefont {C.-H.}\ \bibnamefont {Wang}},\ and\
  \bibinfo {author} {\bibfnamefont {F.}~\bibnamefont {Gao}},\ }\href
  {https://doi.org/10.1088/1751-8121/ababb5} {\bibfield  {journal} {\bibinfo
  {journal} {J. Phys. A: Math. Theor.}\ }\textbf {\bibinfo {volume} {53}},\
  \bibinfo {pages} {385306} (\bibinfo {year} {2020})}\BibitemShut {NoStop}%
\bibitem [{\citenamefont {Bandyopadhyay}\ and\ \citenamefont
  {Russo}(2021)}]{Bandyopadhyay21}%
  \BibitemOpen
  \bibfield  {author} {\bibinfo {author} {\bibfnamefont {S.}~\bibnamefont
  {Bandyopadhyay}}\ and\ \bibinfo {author} {\bibfnamefont {V.}~\bibnamefont
  {Russo}},\ }\href {https://doi.org/10.1103/PhysRevA.104.032429} {\bibfield
  {journal} {\bibinfo  {journal} {Phys. Rev. A}\ }\textbf {\bibinfo {volume}
  {104}},\ \bibinfo {pages} {032429} (\bibinfo {year} {2021})}\BibitemShut
  {NoStop}%
\bibitem [{\citenamefont {Lovitz}\ and\ \citenamefont
  {Johnston}(2022)}]{Lovitz22}%
  \BibitemOpen
  \bibfield  {author} {\bibinfo {author} {\bibfnamefont {B.}~\bibnamefont
  {Lovitz}}\ and\ \bibinfo {author} {\bibfnamefont {N.}~\bibnamefont
  {Johnston}},\ }\href {https://doi.org/10.22331/q-2022-07-07-760} {\bibfield
  {journal} {\bibinfo  {journal} {{Quantum}}\ }\textbf {\bibinfo {volume}
  {6}},\ \bibinfo {pages} {760} (\bibinfo {year} {2022})}\BibitemShut {NoStop}%
\bibitem [{\citenamefont {Bandyopadhyay}\ and\ \citenamefont
  {Russo}(2024)}]{Bandyopadhyay24}%
  \BibitemOpen
  \bibfield  {author} {\bibinfo {author} {\bibfnamefont {S.}~\bibnamefont
  {Bandyopadhyay}}\ and\ \bibinfo {author} {\bibfnamefont {V.}~\bibnamefont
  {Russo}},\ }\href {https://doi.org/10.1103/PhysRevA.110.042406} {\bibfield
  {journal} {\bibinfo  {journal} {Phys. Rev. A}\ }\textbf {\bibinfo {volume}
  {110}},\ \bibinfo {pages} {042406} (\bibinfo {year} {2024})}\BibitemShut
  {NoStop}%
\bibitem [{\citenamefont {Zhu}\ \emph {et~al.}(2025)\citenamefont {Zhu},
  \citenamefont {Zhu}, \citenamefont {Liu},\ and\ \citenamefont
  {Wang}}]{Zhu25}%
  \BibitemOpen
  \bibfield  {author} {\bibinfo {author} {\bibfnamefont {C.}~\bibnamefont
  {Zhu}}, \bibinfo {author} {\bibfnamefont {C.}~\bibnamefont {Zhu}}, \bibinfo
  {author} {\bibfnamefont {Z.}~\bibnamefont {Liu}},\ and\ \bibinfo {author}
  {\bibfnamefont {X.}~\bibnamefont {Wang}},\ }\href
  {https://doi.org/10.1109/TIT.2025.3532701} {\bibfield  {journal} {\bibinfo
  {journal} {IEEE Trans. Inf. Theory}\ }\textbf {\bibinfo {volume} {71}},\
  \bibinfo {pages} {2826} (\bibinfo {year} {2025})}\BibitemShut {NoStop}%
\bibitem [{\citenamefont {Lim}\ \emph {et~al.}(2025)\citenamefont {Lim},
  \citenamefont {Hhan},\ and\ \citenamefont {Kwon}}]{Lim25}%
  \BibitemOpen
  \bibfield  {author} {\bibinfo {author} {\bibfnamefont {Y.}~\bibnamefont
  {Lim}}, \bibinfo {author} {\bibfnamefont {M.}~\bibnamefont {Hhan}},\ and\
  \bibinfo {author} {\bibfnamefont {H.}~\bibnamefont {Kwon}},\ }\href
  {https://doi.org/10.1088/2058-9565/adc034} {\bibfield  {journal} {\bibinfo
  {journal} {Quantum Sci. Technol.}\ }\textbf {\bibinfo {volume} {10}},\
  \bibinfo {pages} {025048} (\bibinfo {year} {2025})}\BibitemShut {NoStop}%
\bibitem [{\citenamefont {Cohen}(2023)}]{Cohen23}%
  \BibitemOpen
  \bibfield  {author} {\bibinfo {author} {\bibfnamefont {S.~M.}\ \bibnamefont
  {Cohen}},\ }\href {https://doi.org/10.1103/PhysRevA.107.012401} {\bibfield
  {journal} {\bibinfo  {journal} {Phys. Rev. A}\ }\textbf {\bibinfo {volume}
  {107}},\ \bibinfo {pages} {012401} (\bibinfo {year} {2023})}\BibitemShut
  {NoStop}%
\bibitem [{\citenamefont {Li}\ \emph {et~al.}(2017)\citenamefont {Li},
  \citenamefont {Wang},\ and\ \citenamefont {Duan}}]{Li17}%
  \BibitemOpen
  \bibfield  {author} {\bibinfo {author} {\bibfnamefont {Y.}~\bibnamefont
  {Li}}, \bibinfo {author} {\bibfnamefont {X.}~\bibnamefont {Wang}},\ and\
  \bibinfo {author} {\bibfnamefont {R.}~\bibnamefont {Duan}},\ }\href
  {https://doi.org/10.1103/PhysRevA.95.052346} {\bibfield  {journal} {\bibinfo
  {journal} {Phys. Rev. A}\ }\textbf {\bibinfo {volume} {95}},\ \bibinfo
  {pages} {052346} (\bibinfo {year} {2017})}\BibitemShut {NoStop}%
\bibitem [{\citenamefont {Halder}\ and\ \citenamefont
  {Streltsov}(2024)}]{Halder24}%
  \BibitemOpen
  \bibfield  {author} {\bibinfo {author} {\bibfnamefont {S.}~\bibnamefont
  {Halder}}\ and\ \bibinfo {author} {\bibfnamefont {A.}~\bibnamefont
  {Streltsov}},\ }\href
  {https://www.rintonpress.com/xxqic24/qic-24-1314/1081-1098.pdf} {\bibfield
  {journal} {\bibinfo  {journal} {Quantum Inf. Comput.}\ }\textbf {\bibinfo
  {volume} {24}},\ \bibinfo {pages} {1081 } (\bibinfo {year}
  {2024})}\BibitemShut {NoStop}%
\bibitem [{\citenamefont {Fu}\ \emph {et~al.}(2014)\citenamefont {Fu},
  \citenamefont {Leung},\ and\ \citenamefont {Man\ifmmode~\check{c}\else
  \v{c}\fi{}inska}}]{Fu14}%
  \BibitemOpen
  \bibfield  {author} {\bibinfo {author} {\bibfnamefont {H.}~\bibnamefont
  {Fu}}, \bibinfo {author} {\bibfnamefont {D.}~\bibnamefont {Leung}},\ and\
  \bibinfo {author} {\bibfnamefont {L.}~\bibnamefont
  {Man\ifmmode~\check{c}\else \v{c}\fi{}inska}},\ }\href
  {https://doi.org/10.1103/PhysRevA.89.052310} {\bibfield  {journal} {\bibinfo
  {journal} {Phys. Rev. A}\ }\textbf {\bibinfo {volume} {89}},\ \bibinfo
  {pages} {052310} (\bibinfo {year} {2014})}\BibitemShut {NoStop}%
\bibitem [{\citenamefont {Cohen}(2022)}]{Cohen22}%
  \BibitemOpen
  \bibfield  {author} {\bibinfo {author} {\bibfnamefont {S.~M.}\ \bibnamefont
  {Cohen}},\ }\href {https://doi.org/10.1103/PhysRevA.105.022207} {\bibfield
  {journal} {\bibinfo  {journal} {Phys. Rev. A}\ }\textbf {\bibinfo {volume}
  {105}},\ \bibinfo {pages} {022207} (\bibinfo {year} {2022})}\BibitemShut
  {NoStop}%
\bibitem [{\citenamefont {Vidal}(2000)}]{Vidal00}%
  \BibitemOpen
  \bibfield  {author} {\bibinfo {author} {\bibfnamefont {G.}~\bibnamefont
  {Vidal}},\ }\href {https://doi.org/10.1080/09500340008244048} {\bibfield
  {journal} {\bibinfo  {journal} {J. Mod. Opt.}\ }\textbf {\bibinfo {volume}
  {47}},\ \bibinfo {pages} {355} (\bibinfo {year} {2000})}\BibitemShut
  {NoStop}%
\bibitem [{\citenamefont {Lo}\ and\ \citenamefont {Popescu}(2001)}]{Lo01}%
  \BibitemOpen
  \bibfield  {author} {\bibinfo {author} {\bibfnamefont {H.-K.}\ \bibnamefont
  {Lo}}\ and\ \bibinfo {author} {\bibfnamefont {S.}~\bibnamefont {Popescu}},\
  }\href {https://doi.org/10.1103/PhysRevA.63.022301} {\bibfield  {journal}
  {\bibinfo  {journal} {Phys. Rev. A}\ }\textbf {\bibinfo {volume} {63}},\
  \bibinfo {pages} {022301} (\bibinfo {year} {2001})}\BibitemShut {NoStop}%
\bibitem [{\citenamefont {Eisert}\ and\ \citenamefont
  {Briegel}(2001)}]{Eisert01}%
  \BibitemOpen
  \bibfield  {author} {\bibinfo {author} {\bibfnamefont {J.}~\bibnamefont
  {Eisert}}\ and\ \bibinfo {author} {\bibfnamefont {H.~J.}\ \bibnamefont
  {Briegel}},\ }\href {https://doi.org/10.1103/PhysRevA.64.022306} {\bibfield
  {journal} {\bibinfo  {journal} {Phys. Rev. A}\ }\textbf {\bibinfo {volume}
  {64}},\ \bibinfo {pages} {022306} (\bibinfo {year} {2001})}\BibitemShut
  {NoStop}%
\bibitem [{\citenamefont {Goswami}\ and\ \citenamefont
  {Halder}(2023)}]{Goswami23}%
  \BibitemOpen
  \bibfield  {author} {\bibinfo {author} {\bibfnamefont {S.}~\bibnamefont
  {Goswami}}\ and\ \bibinfo {author} {\bibfnamefont {S.}~\bibnamefont
  {Halder}},\ }\href {https://doi.org/10.1103/PhysRevA.108.012405} {\bibfield
  {journal} {\bibinfo  {journal} {Phys. Rev. A}\ }\textbf {\bibinfo {volume}
  {108}},\ \bibinfo {pages} {012405} (\bibinfo {year} {2023})}\BibitemShut
  {NoStop}%
\bibitem [{\citenamefont {Zhang}\ \emph {et~al.}(2023)\citenamefont {Zhang},
  \citenamefont {Luo}, \citenamefont {Yu},\ and\ \citenamefont
  {Zhou}}]{Zhang23}%
  \BibitemOpen
  \bibfield  {author} {\bibinfo {author} {\bibfnamefont {X.}~\bibnamefont
  {Zhang}}, \bibinfo {author} {\bibfnamefont {M.}~\bibnamefont {Luo}}, \bibinfo
  {author} {\bibfnamefont {G.}~\bibnamefont {Yu}},\ and\ \bibinfo {author}
  {\bibfnamefont {X.}~\bibnamefont {Zhou}},\ }\href
  {https://doi.org/10.1103/PhysRevA.108.032215} {\bibfield  {journal} {\bibinfo
   {journal} {Phys. Rev. A}\ }\textbf {\bibinfo {volume} {108}},\ \bibinfo
  {pages} {032215} (\bibinfo {year} {2023})}\BibitemShut {NoStop}%
\end{thebibliography}%
\end{document}